%% file: LocalInfo.tex
\documentclass[11pt]{article}

\usepackage{algorithmic,algorithm}
\usepackage{color}
\usepackage{amsmath}
\usepackage{amsthm}
\usepackage{amssymb}
\usepackage{float}
\usepackage{cite}

\begin{document}


\title{The Power of Local Information in Social Networks}

\date{}
\author{Christian Borgs \thanks{Microsoft Research New England. Email: borgs@microsoft.com.}
\and  Michael Brautbar\thanks{Department of Computer and Information Science, University of Pennsylvania. Email: brautbar@cis.upenn.edu.} \and  Jennifer Chayes\thanks{Microsoft Research New England. Email: jchayes@microsoft.com.} 
\and  Sanjeev Khanna \thanks{Department of Computer and Information Science, University of Pennsylvania. Email: sanjeev@cis.upenn.edu.}
\and Brendan Lucier\thanks{Microsoft Research New England. Email: brlucier@microsoft.com.}
}

\maketitle


\newcommand{\graph}           {\ensuremath{{\cal G}} }
\newcommand{\othergraph}      {\ensuremath{{\cal H}} }
\newcommand{\util}[1]         {\text{\it utility}(#1)}
\newcommand{\idealized}[2]    {{\text{\it ideal}}_{#1}(#2)}
\newcommand{\totaldeg}[1]     {\text{\it deg}(#1)}
\newcommand{\indeg}[2]        {\text{\it in-deg}_{#2}(#1)}
\newcommand{\outdeg}[1]       {\text{\it out-deg}(#1)}
\newcommand{\inset}[1]        {\ensuremath{\text{\it I}_{#1}}}
\newcommand{\outset}[1]       {\ensuremath{\text{\it O}_{#1}}}
\newcommand{\totaltriang}[1]  {\ensuremath{\Delta(#1)}}
\newcommand{\intriang}[1]     {\ensuremath{\Delta_{I}(#1)}}

\newtheorem{theorem}{Theorem}[section]
\newtheorem{corollary}[theorem]{Corollary}
\newtheorem{lemma}[theorem]{Lemma}
\newtheorem{claim}[theorem]{Claim}
\newtheorem{fact}[theorem]{Fact}
\newtheorem{proposition}[theorem]{Proposition}
\newtheorem{conjecture}{Conjecture}
\newtheorem{property}{Property}
\newtheorem{observation}[theorem]{Observation}
\newtheorem{remark}{Remark}
\newtheorem{definition}{Definition}

\newenvironment{proofof}[1]{\begin{proof}[Proof of #1]}{\end{proof}}

\floatstyle{ruled}
\newfloat{algorithm}{tbp}{loa}
\floatname{algorithm}{Algorithm}

\newif\iftechnicalreport
\technicalreportfalse
\technicalreporttrue

\newcommand{\ignore}[1]{}
\newcommand{\comment}[1]{{\bf\em$<$#1$>$}}
\newcommand{\brendan}[1]{{\color{blue}{BJL: #1}}}
\newcommand{\mickey}[1]{{\color{red}{MB: #1}}}

\newcommand{\be}{\begin{equation}}
\newcommand{\ee}{\end{equation}}
\newcommand{\req}[1]{(\ref{#1})}
\newcommand{\argmin}{\mathop{\rm argmin}}
\newcommand{\vr}[1]{{\mathbf{#1}}}
\newcommand{\bydef}{\stackrel{\bigtriangleup}{=}}
\newcommand{\eps}{\epsilon}
\newcommand{\argmax}{\mathop{\rm argmax}}
\newcommand{\mm}[1]{\mathrm{#1}}
\newcommand{\mc}[1]{\mathcal{#1}}

\def \ONE {{ \mathbf 1 }}
\def \EE   {{\mathbb E}}
\def \E   {{\mathbb E}}
\def \Prob   {{\mathbb P}}
\def \OPT {\mathcal{OPT}}
\def \vf  {\textrm{vf}}
\def \dvf {\varphi}
\def \utility {u}



\begin{abstract}
\input{abstract.tex}

\end{abstract}

\input{LocalInfo-intro.tex}
\input{LocalInfo-Prelim.tex}

\input{LocalInfo-PA.tex}

\input{LocalInfo-mindom.tex}
\input{LocalInfo-concl.tex}

\section*{Acknowledgments}
The author Sanjeev Khanna was supported in part by NSF awards CCF-1116961 and IIS-0904314, and by ONR MURI grant N00014-08-1-0747.

\bibliographystyle{plain}
\bibliography{LocalInfo}

\appendix
\section*{APPENDIX}

\input{LocalInfo-appendix-PA.tex}

\input{LocalInfo-appendix-mindom.tex}

\input{LocalInfo-appendix-bounds.tex}

\end{document}

%% file: abstract.tex
We study the power of \textit{local information algorithms} for optimization problems on social and technological networks.  We focus on sequential algorithms where the network topology is initially unknown and is revealed only within a local neighborhood of vertices that have been irrevocably added to the output set.
This framework models the behavior 
of an external agent that does not have direct access to the network data,
such as a user interacting with an online social network.

We study a range of problems under this model of algorithms with local information.  
When the underlying graph is a preferential attachment network, we show that 
one can find the root (i.e.\ initial node) in a polylogarithmic number of steps, using 
a local algorithm that repeatedly queries the visible node of maximum degree.  
This addresses an open question of Bollob{\'a}s and Riordan.
This result is motivated by its implications: 
we obtain polylogarithmic approximations to problems such as finding the smallest subgraph that connects a subset of nodes, finding the highest-degree nodes, and finding a subgraph that maximizes vertex coverage per subgraph size.

Motivated by problems faced by recruiters in online networks, 
we also consider network coverage problems on arbitrary graphs.
We demonstrate a sharp threshold on the level of visibility required: at a certain visibility level it is possible to design algorithms that nearly match the best approximation possible even with full access to the graph structure, but with any less information it is impossible to achieve a non-trivial approximation.
We conclude that a network provider's decision of how much structure to make visible to its users can have a significant effect on a user's ability to interact strategically with the network.

%% file: LocalInfo-intro.tex
\section{Introduction}

In the past decade there has been a surge of interest in the nature of complex networks that arise in social and technological contexts; see \cite{EK10} for a recent survey of the topic.
In the computer science community, this attention has been directed largely towards algorithmic issues, such as the extent to which network structure can be leveraged into efficient methods for solving complex tasks.  Common problems include finding influential individuals, detecting communities, constructing subgraphs with desirable connectivity properties, and so on.  

The standard paradigm in these settings is that an algorithm
has full access to the network graph structure.  
More recently there has been growing interest in \emph{local} algorithms, 
in which decisions are based upon local rather than global network structure.  This locality of computation has been motivated by applications to distributed algorithms \cite{NaorS93,GiakkoupisS12}, improved runtime efficiency \cite{FaloutsosMT04,SpielmanTeng}, and property testing \cite{HassidimKNO09,RubSha12}.  In this work we consider a different motivation: 
in some circumstances, an optimization is performed by an external user who has inherently restricted visibility of the network topology.

For such a user, the graph structure is revealed incrementally within a local neighborhood of nodes for which a connection cost has been paid.  The use of local algorithms in this setting is necessitated by constraints on network visibility, rather than being a means toward an end goal of efficiency or parallelizability.

As one motivating example, consider an agent in a social network who wishes to find (and link to) a highly connected individual.  For instance, this agent may be a newcomer to a community 
(such as an online gaming or niche-based community) wanting to interact with influential or popular individuals.
Finding a high-degree node is a straightforward algorithmic problem without information constraints, but many online and real-world social networks reveal graph structure only within one or two hops from a user's existing connections.  

In another motivating example, consider the problem faced by a recruiter attempting to build capital in the form of a network of contacts.  Employment-focused social network applications such as LinkedIn emphasize the importance of having many nodes within few hops of one's neighborhood of connections.  
For instance, a user can browse and perform searchers over individuals that are within network distance $2$ of themselves (i.e.\ friends of friends).  Since it is socially and timely costly to create connections, a recruiter is motivated to befriend only a small number of individuals that assist in covering as much of the network as possible.
Is it possible for an agent to solve such a problem given information only about local graph structure, i.e. what is revealed on an online networking site? This question is relevant not only for individual users, but also to the designer of a social networking service who must decide how much topological information to reveal. 


More generally, we consider graph algorithms in a setting of restricted network visibility.  We focus on optimization problems for which the goal is to return a subset of the nodes in the network; this includes coverage, connectivity, and search problems.  An algorithm in our framework proceeds by incrementally and adaptively building an output set of nodes, corresponding to those vertices of the graph that have been queried (or connected to) so far.  When the algorithm has queried a set $S$ of nodes, the structure of the graph within a small radius of $S$ is revealed, guiding future queries.  The principle challenge in designing such an algorithm is that decisions must be based solely on local information, whereas the problem to be solved may depend on the global structure of the graph.
In addition to these restrictions, we ask for algorithms that run in polynomial time.

For some node coverage problems, such as the recruitment problem described above (i.e.\ finding a minimum dominating set), we show how to design local information algorithms for arbitrary networks whose performances approximately match the best possible even when information about network structure is unrestricted.  We also demonstrate that the amount of local information available is of critical importance: strong positive results are possible at a certain range of visibility (made explicit below), but non-trivial algorithms become impossible when less information is made available.  This observation has implications for the design of online networks, such as the amount of information to provide a user about the local topology: seemingly arbitrary design decisions may have a significant impact on a user's ability to interact with the network. 

For other problems, such as finding and linking to the most highly connected individual in a social network, we derive strong lower bounds on the performance of local algorithms for general graphs.  
We therefore turn to the class of preferential attachment (PA) graphs, which model 
properties of many real-world social and technological networks. 

For PA networks, we prove that local information algorithms do 
surprisingly well at many optimization problems, including finding the $k$ vertices of highest degree and shortest path routing (up to polylogarithmic factors).  The structure of social networks can therefore play an important role in determining the feasibility of solving an optimization problem under local information constraints.

\paragraph{Results and Techniques}

Our first set of results concerns algorithms for preferential attachment (PA) networks.  Such networks are defined by a process by which nodes are added sequentially and form random connections to existing nodes, where the probability of connecting to a node is proportional to its degree.  

We first consider the problem of finding the root (first) node in a PA network.  A random walk would encounter the root in $\tilde{O}(\sqrt{n})$ steps (where $n$ is the number of nodes in the network).  The question of whether a better local information algorithm exists for this problem was posed by Bollob{\'a}s and Riordan \cite{BR-01}, who state that ``an interesting question is whether a short path between two given vertices can be constructed quickly using only `local' information" \cite{BR-01}.  
They conjecture that such short paths can be found locally in $\Theta(\log n)$ steps.  We make the first progress towards this conjecture by showing that polylogarithmic time is sufficient: there is an algorithm that finds the root of a PA network in $O(\log^4(n))$ time, with high probability.  We show how to use this algorithm to obtain polylogarithmic approximations for finding the smallest subgraph that connects a subset of nodes (including shortest path), finding the highest-degree nodes, and finding a subgraph that maximizes vertex coverage per subgraph size.
 
The local information algorithm we propose uses a natural greedy approach: at each step, it queries the visible node with highest degree.  Demonstrating that such an algorithm reaches the root in $O(\log^4(n))$ steps requires a probabilistic analysis of the PA process.  A natural intuition is that the greedy algorithm will find nodes of ever higher degrees over time.  However, such progress is impeded by the presence of high-degree nodes with only low-degree neighbors.  We show that these bottlenecks are infrequent enough that they do not significantly hamper the algorithm's progress.  To this end, we derive a connection between node degree correlations and supercritical branching processes to prove that a path of high-degree vertices leading to the root is always available to the algorithm. 

Motivated by job recruiting in networks, we then explore local information algorithms for dominating set and coverage problems on general graphs.  
%
A dominating set is a set $S$ such that each node in the network is either in $S$ or the neighborhood of $S$.  We design a randomized local information algorithm for the minimum dominating set problem that achieves an approximation ratio that nearly matches the lower bound on polytime algorithms with no information restriction.  As has been noted in \cite{GK-98}, the greedy algorithm that repeatedly selects the visible node that maximizes the size of the dominated set can achieve a very bad approximation factor.  We consider a modification of the greedy algorithm: after each greedy addition of a new node $v$, the algorithm will also add a random neighbor of $v$.  We show that this randomized algorithm obtains an approximation factor that matches the known lower bound of $\Omega(\log \Delta)$ (where $\Delta$ is the maximum degree in the network) up to a constant factor.  We also show that having enough local information to choose the node that maximizes the incremental benefit to the dominating set size is crucial: any algorithm that can see only the degrees of the neighbors of $S$ would achieve a poor approximation factor of $\Omega(n)$ .

Finally, we extend these results to related coverage problems.  For the partial dominating set problem (where the goal is to cover a given constant fraction of the network with as few nodes as possible) we give an impossibility result: no local information algorithm can obtain an approximation better than $O(\sqrt{n})$ on networks with $n$ nodes.  However, a slight modification to the local information algorithm for minimum dominating set yields a a bicriteria result: given $\epsilon > 0$, we compare the performance of an algorithm that covers $\rho n$ nodes with the optimal solution that covers $\rho (1+\epsilon) n$ nodes (assuming $\rho(1+\epsilon) \leq 1$).  Our modified algorithm achieves a $O((\rho\epsilon)^{-1}H(\Delta))$ approximation.

 (in which we compare performance against an adversary who must cover an additional $\epsilon$ fraction of the network).  We also consider the ``neighbor-collecting'' problem, in which the goal is to minimize $c|S|$ plus the number of nodes left undominated by $S$, for a given parameter $c$.  For this problem we show that the minimum dominating set algorithm yields an $O(c\log \Delta)$ approximation (where $\Delta$ is the maximum degree in the network), and that the dependence on $c$ is unavoidable.

\paragraph{Related Work}
Over the last decade there has been a substantial body of work on understanding the 
power of sublinear-time approximations.
In the context of graphs, the goal is to understand 
how well one can approximate graph properties in a sublinear number of queries. 
See~\cite{RubSha12} and~\cite{Goldreich10c} for recent surveys. 
In the context of social networks a recent work has suggested 
the Jump and Crawl model, where algorithms have no direct access to the network but can either sample a node uniformly (Jump) or access a neighbor of a previously discovered node (a Crawl)~\cite{BK-10}.  Local information algorithms can be thought of as generalizing the Jump and Crawl query framework to include an informational dimension. A Crawl query will now return any node in the local neighborhood of nodes seen so far while Jump queries would allow access to unexplored regions of the network.

Motivated by distributed computation, a notion of local computation was formalized by~\cite{RubinfeldTVX11} and further developed in~\cite{AlonRVX12}.
The goal of a local computation algorithm is to compute only certain specified bits of a global solution. 
In contrast, our notion of locality is motivated by informational rather than computational constraints imposed upon a sequential algorithm. As a result, the informational dimension of visibility tends not to play a role in the analysis of local computation algorithms. In contrast, the main issue in designing a local information algorithm is to find a good tradeoff between the amount of local information the algorithm is allowed to see to the number of queries it needs to make in order to solve a network optimization problem. 


Local algorithms motivated by efficient computation, rather than informational constraints, were explored by~\cite{AndersenCL06,SpielmanTeng}. These works explore local approximation of graph partitions to efficiently find a global solution.  
In particular, they explore the ability to find a cluster containing a given vertex by querying only close-by nodes.

Preferential attachment (PA) networks were suggested by~\cite{BA-99} as a model for large social networks.
There has been much work studying the properties of such networks, such as degree distribution \cite{BR2-01} and diameter \cite{BR-01}; see~\cite{Bollo03} for a short survey.
The problem of finding high degree nodes, using only uniform sampling and local neighbor queries, is explored in \cite{BK-10}. The question of whether a polylogarithmic time Jump and Crawl algorithm exists for finding a high degree node in preferential attachment graphs was left open therein. 

The low diameter of PA graphs can be used to implement distributed algorithms in which nodes repeatedly broadcast information to their neighbors
\cite{GiakkoupisS12,DFF-11}. A recent work~\cite{DFF-11} showed that such algorithms can be used for fast rumor spreading. 
Our results on the ability to find short paths in such graphs differs in that our algorithms are sequential, with a small number of queries, rather than applying broadcast techniques.

The ability to quickly find short paths in social networks has been the focus of much study, especially in the context of small-world graphs \cite{Kleinberg00,GiakkoupisS11}.
It is known that local routing using short paths is possible in such models, 
given some awareness of global network structure (such as coordinates in an underlying grid).

In contrast, our shortest-path algorithm for PA graphs
does not require an individual know the graph structure beyond the degrees of his neighbors.
However, our result requires that routing can be done from both endpoints; in other words, both nodes are trying to find each other.

For the minimum dominating set problem, Guha and Khuller \cite{GK-98} designed a local $O(\log \Delta)$ approximation algorithm (where $\Delta$ is the maximum degree in the network).  
As a local information algorithm, their method requires that the network structure is revealed up to distance two from the current dominating set.  By contrast, our local information algorithm 
requires less information to be revealed on each step.
Our focus, and the motivation behind this distinction, is to determine sharp bounds on the amount of local information required to approximate this problem (and others) effectively.

%% file: LocalInfo-Prelim.tex
\section{Model and Preliminaries}
\label{section.Model}

\paragraph{Graph Notation}
We write $G = (V,E)$ for an undirected graph with node and edge sets $V$ and $E$, respectively.  We write $n_G$ for the number of nodes in $G$, $d_G(v)$ for the degree of a vertex $v$ in $G$, and $N_G(v)$ for the set of neighbors of $v$.  Given a subset of vertices $S \subseteq V$, $N_G(S)$ is the set of nodes adjacent to at least one node in $S$.  We also write $D_G(S)$ for the set of nodes \emph{dominated} by $S$: $D_G(S) = N_G(S) \cup S$.  We say $S$ is a \emph{dominating set} if $D_G(S) = V$.  Given nodes $u$ and $v$, the distance between $u$ and $v$ is the number of edges in the shortest path between $u$ and $v$.  The distance between vertex sets $S$ and $T$ is the minimum distance between a node in $S$ and a node in $T$.  Given a subset $S$ of nodes of $G$, the subgraph induced by $S$ is the subgraph consisting of $S$ and every edge with both endpoints in $S$.  Finally, $\Delta_G$ is the maximum degree in $G$.  In all of the above notation we often suppress the dependency on $G$ when clear from context.

\paragraph{Algorithmic Framework}
We focus on graph optimization problems in which the goal is to return a minimal-cost\footnote{In most of the problems we consider, the cost of set $S$ will simply be $|S|$.} set of vertices $S$ satisfying a feasibility constraint.  
We will consider a class of algorithms that build $S$ incrementally under local information constraints.  
We begin with a definition of distance to sets and a definition local neighborhoods.

\begin{definition}
The distance of $v$ from a set $S$ of nodes in a graph $G$ is the minimum, over all nodes $u \in S$, of the shortest path length from $v$ to $u$ in $G$.
\end{definition}

\begin{definition}[Local Neighborhood]
Given a set of nodes $S$ in the graph $G$, the \emph{$r$-closed neighborhood around $S$} is 
the induced subgraph of $G$ containing all nodes at distance less than or equal to $r$ from $S$, plus the degree of each node at distance $r$ from $S$. The \emph{$r$-open neighborhood around $S$} is the $r$-closed neighborhood around $S$, after the removal of all edges between nodes at distance exactly $r$ from $S$.
\end{definition} 

\begin{definition}[\textbf{Local Information Algorithm}]
Let $G$ be an undirected graph unknown to the algorithm, where each vertex is assigned a unique identifier.  
For integer $r \ge 1$, a (possibly randomized) algorithm is an \textit{$r$-local algorithm} if: 
\begin{enumerate}
\item The algorithm proceeds sequentially, growing step-by-step a set $S$ of nodes, where $S$ is initialized either to $\emptyset$ or to some seed node.
\item
Given that the algorithm has queried a set $S$ of nodes so far, it can only observe the $r$-open neighborhood around $S$.
\item
On each step, the algorithm can add a node to $S$ either by selecting a specified vertex from the $r$-open neighborhood around $S$ (a \emph{crawl}) or by selecting a vertex chosen uniformly at random from all graph nodes (a \emph{jump}).  
\item In its last step the algorithm returns the set $S$ as its output.
\end{enumerate}
\end{definition}

Similarly, for $r \ge 1$, we call an algorithm a \textit{$r^{+}$-local algorithm} if
its local information (i.e.\ in item $2$) is made from the $r$-closed neighborhood around $S$.

We focus on computationally efficient (i.e.\ polytime) local algorithms.  
Our framework applies most naturally to coverage, search, and connectivity problems, where the family of valid solutions is upward-closed.  

More generally, it is suitable for measuring the complexity, using only local information, for finding a subset of nodes having a desirable property.  In this case the size of $S$ measures the number of queries made by the algorithm; we think of the graph structure revealed to the algorithm as having been paid for by the cost of $S$.

For our lower bound results, we will sometimes compare the performance of an $r$-local algorithm with that of a (possibly randomized) algorithm that is also limited to using Jump and Crawl queries, but may use full knowledge of the network topology to guide its query decisions. The purpose of such comparisons is to emphasize instances where it is the lack of information about the network structure, rather than the necessity of building the output in a local manner, that impedes an algorithm's ability to perform an optimization task.

%% file: LocalInfo-PA.tex
\section{Preferential Attachment Graphs}
\label{section.PAmain}
 
We now focus our attention on algorithms for graphs generated by the preferential attachment (PA) process, 
conceived by Barab\'{a}si and Albert \cite{BA-99}. Informally, the process is defined sequentially
with nodes added one by one.  When a node is added it sends $m$ links backward to previously created nodes, connecting to  a node with probability proportional to its current degree. 

We will use the following, now standard, formal definition of the process, due to \cite{BR-01}.
Given $m \geq 1$, we inductively define random graphs $G_m^{t}$, $1 \le t \le n$.
The vertex set for $G_m^t$ is $[t]$ where each node is identified by its creation time (index). $G_m^{1}$ is the graph with node $1$ and 
$m$ self-loops. Given $G_m^{(t-1)}$, form $G_m^{t}$
by adding node $t$ and then forming $m$ edges from $t$ to nodes in $[t]$,
say $p_1(t), \dotsc, p_m(t)$.  The nodes $p_k(t)$ are referred to as the \emph{parents} of $t$.
The edges are formed sequentially.  
For each $k$, node $s$ is chosen
with probability $\deg(s)/z$ if $s < t$, 
or $(\deg(s)+1)/z$ if $s = t$, where $z$ is a normalization factor.

Note that $\deg(s)$ denotes degree in $G_m^{t-1}$, counting previously-placed edges.

We first present a $1$-local approximation algorithm for the following simple problem on PA graphs: given an arbitrary node $u$, return a minimal connected subgraph containing nodes $u$ and $1$ (i.e.\ the root of $G_m^{n}$).  

Our algorithm, TraverseToTheRoot, is listed as Algorithm \ref{traverse}.  The algorithm grows a set $S$ of nodes by starting with $S = \{u\}$ and then repeatedly adding the node in $N(S)\backslash S$ with highest degree.  
We will show that, with 
high probability, this algorithm traverses the root node within $O(\log^4(n))$ steps.
 
\begin{algorithm}
\caption{TraverseToTheRoot} \label{traverse}
\begin{algorithmic}[1]
\STATE Initialize a list $L$ to contain an arbitrary node $\{u\}$ in the graph.
\WHILE{$L$ does not contain node $1$}
         \STATE Add a node of maximum degree in $N(L) \backslash L$ to $L$.
\ENDWHILE
\STATE return $L$.
\end{algorithmic}
\end{algorithm} 

\begin{theorem}
\label{thm.PA}
With probability $1-o(1)$ over the preferential attachment process on $n$ nodes,
TraverseToTheRoot returns a set of size $O(\log^4(n))$.
\end{theorem}

\paragraph{Remark:}
For convenience, we have defined TraverseToTheRoot assuming that the algorithm can determine when it has successfully traversed the root.  This is not necessary in general; our algorithm will instead have the guarantee that, after $O(\log^4(n))$ steps, it has traversed node $1$ with high probability.

Before proving Theorem \ref{thm.PA}, we discuss its algorithmic implications below.

\subsection{Applications of Fast Traversal to the Root}
\label{section.PAapps}
We now describe how to use TraverseToTheRoot to implement local algorithms for other problems on PA networks.  For ease of readability, the proofs of all auxiliary lemmas will be only presented at the end of the chapter.

\paragraph{s-t connectivity.}
In the $s$-$t$ connectivity (or shortest path) 
problem we are given two nodes $s,t$ in an undirected graph and must find a small connected subgraph that contains $s$ and $t$.  

\begin{corollary}
\label{thm.connectivity}
Let $G$ be a PA graph on $n$ nodes.
Then, with probability $1-o(1)$ over the PA process,
Algorithm \ref{stconnect} (listed above), a $1$-local algorithm, returns a connected subgraph of size $O(\log^4(n))$ containing vertices $s$ and $t$. Furthermore, for any fixed $k$, with probability $1-o(1)$ over the PA process, a subset of $k$ nodes can be connected by a local algorithm in $O(k\log^4(n))$ steps, using a subset of size $O(k\log^4(n))$.

\end{corollary}
\begin{proof}
Theorem~\ref{thm.PA} implies that, with probability $1-o(1)$, algorithm TravesetToTheRoot($G,s$) returns a path from $s$ to node $1$ in time $O(\log^4(n))$. 

Similarly, with probability $1-o(1)$, TraverseToTheRoot($G,t$) returns a connected path from $s$ to node $1$ in time $O(\log^4(n))$. Concatenating the two paths at node $1$ 
is a path of length $O(\log^4(n))$ from $s$ to $t$. Given $k$ terminal one can connect all of them to nodes 1.
Theorem~\ref{thm.PA} implies that for each terminal, with probability $1-o(1)$, algorithm TravesetToTheRoot($G,s$) returns a path from $s$ to node $1$ in time $O(\log^4(n))$. For a fixed $k$, the claim then follows from the union bound.

\end{proof}

\begin{algorithm}
\caption{s-t-Connect} \label{stconnect}
\begin{algorithmic}[1]
\STATE $P_1 \leftarrow $ TraverseToTheRoot($G,s$)
\STATE $P_2 \leftarrow $ TraverseToTheRoot($G,t$)
\STATE Return $P_1 \cup P_2$
\end{algorithmic}
\end{algorithm} 

One may ask whether our results can be extended to arbitrary networks. We next that Corollary \ref{thm.connectivity} does not extend to general graphs: local algorithms 
cannot achieve sublinear approximations. We provide a strong lower bound: a randomized $r$-local algorithm will have approximation factor $\Omega(n/r^2)$, even if it is allowed to fail with probability $1/3$.  Furthermore, such a bound holds even if we ask that the graph be well-connected: the lower bound is $\Omega(n/kr^2)$ for $k$-edge-connected graphs.

\begin{theorem}
\label{lemma.lbconnect}
Let $k$, $r$, $n \ge \text{max} \{ k , r \}$, be positive integers.
For any $r$-local algorithm for the $s$-$t$ connectivity problem with success probability bigger than $\frac{1}{3}$, the expected approximation over successful runs is $\Omega(\frac{n}{kr^{2}})$ for $k$-edge-connected graphs.

\end{theorem}
\begin{proof}
We first focus our attention on proving the claim for $r$-local algorithms on $k=1$ connected graphs. 
The proof will invoke the application of Yao's minmax principle for the performance of Monte Carlo randomized algorithms on a family of inputs~\cite{Yao77}. The lemma states that the expected cost of the optimal 
deterministic Monte Carlo algorithm with failure probability of $\epsilon \in [0,1]$ on an arbitrary distribution over the family of inputs is a lower bound on
the expected cost of the optimal randomized Monte Carlo algorithm with failure probability of $\frac{\epsilon}{2}$ over that family of inputs. 

To use the lemma we will focus on Monte Carlo deterministic algorithms that have failure probability smaller than
some small constant, say $\frac{1}{3}$,
and analyze their performance on a uniformly at random chosen input from a family of inputs constructed below.
Given $n$ we construct the family of inputs. Each input is 
a graphs constructed as follows: we have two distinct nodes $s$ and $t$.
We define a broken path as path on $2r+4$ nodes where the 'middle' edge has been removed.
The graph will be made from $\frac{n-2-(2r+4)}{2r+4}$ distinct broken paths from $s$ to $t$
together with one distinct connected path connecting $s$ to $t$. The identity of the connected path would be chosen uniformly at random from the set of $\frac{n-2}{2r+4}$ paths.
In total the family of inputs contains $\frac{n-2}{2r+4}$ members.
As the algorithm is $r$-local, being at $s$ or $t$ it cannot see the middle node on a broken path
so it cannot decide if a path is broken before traversing at least one node in it.
A compelling property therefore holds: if the algorithm has not found the connected path after $i$ queries then 
the algorithm learns nothing about the identity of the broken path except that it is not one of the paths it traversed so far.
As the connected path is chosen uniformly at random from all paths, the probability that after $\frac{1}{2}\frac{n-2}{2r+4}$ queries the connected path is found is at most $\frac{1}{2}$. Thus,
conditioned on the algorithm being successful (an event having probability at least $\frac{1}{4}$),
the expected cost of finding a path from $s$ to $t$ must be $\Omega(\frac{n}{r})$. Using Yao's principle applied to Monte Carlo algorithms the worst case expected cost of a randomized algorithm on at least one of the inputs would be at least $\Omega(\frac{n}{r})$. However, on any of the inputs, an algorithm with full knowledge of the graph can find the connected path in any graph in the family
in $2r+4$ queries. The approximation ratio of the $r$-local algorithm would therefore be worse than
$\Omega(\frac{n}{r^{2}})$.

It is not hard to generalize the construction to $k$ connected graphs by replacing parts of each path in the construction by a complete graph on that nodes. For a detailed description see below.

We create $\frac{n-2}{k(2r+4)}$ distinct paths, connecting $s$ to $t$, each on
$k(2r+4)$ nodes.
We choose all but one of them to be broken.
For a given path, the node at distance $k(2r+2)$ is chosen to be broken.
In each broken path we form a clique between every consecutive $k$ nodes, starting from $s$, up to the point the path is broken at. If we denote the path by $p_1=s,p_2,\ldots,p_{2r+4}=t$, then
the first clique contains nodes $p_1=s,p_2, \ldots, p_k$ and the last is on $p_{\frac{r}{2}-k+1},  p_{\frac{r}{2}-k+2},\ldots ,p_{\frac{r}{2}}$. 
Similarly, we form cliques on the nodes on $t$ side of the broken path on every $k$ consecutive nodes, starting at node $p_{\frac{r}{2}}+1$.
The graph then becomes $k$-edge connected, and we can repeat the argument given above.
\end{proof}

\paragraph{Finding high degree nodes.}
A natural question on graphs is to find a node with maximal degree.
As we now show, the algorithm TraverseToTheRoot obtains a polylogarithmic approximation to this problem.
  
\begin{corollary}
\label{thm.highdeg}
Let $G$ be a preferential attachment graph on $n$ nodes. 
Then, with probability $1-o(1)$, 

\begin{itemize}
\item Algorithm TraverseToTheRoot will return a node of degree at least $\frac{1}{\log^2(n)}$ of the maximum degree in the graph, in time $O(\log^4(n))$. 
\item For any fixed $k$, algorithm TraverseToTheRoot can be extended to return, with probability $1-o(1)$ and in time $O(\log^4(n))$, a set of $k$ nodes which provides an $O(1/\log^3(n))$ approximation to sum of the $k$ largest degrees in the network.
\end{itemize}

\end{corollary}

\begin{proof}

We first prove part 1 of the corollary. TraveseToTheRoot ends when node $1$ is found.  In Appendix \ref{sec.node1degree} we prove that,
with probability $1-o(1)$, node $1$ has degree at least 
$\frac{m\sqrt{n}}{\log(n)}$.
However, from \cite{Bollo03} (Theorem 17), with probability $1-o(1)$, the maximum degree is less than $m\sqrt{n} \log(n)$. As TraveseToTheRoot runs, with high probability, in $O(\log^4(n))$ steps we conclude that 
a node of degree at least $\frac{1}{\log^2(n)}$ times the maximum degree in the graph is found in time $O(\log^4(n))$. 

To prove the second part of the corollary, we note that by letting the algorithm run until it finds node $1$, then continue for an additional $k$ steps it achieves a solution that is at least as good as returning the nodes indexed $1$ through $k$. In other words,
the $i$'th node in the solution set $S$ returned by the algorithm has degree at least as high as that node $i$.
This follows from a simple induction. As proven in Appendix \ref{sec.node1degree}, with probability $1-o(1)$, the node $i$ has degree at least $\Omega(\frac{\sqrt{n}}{\log^2n})$. As the maximum degree (and hence the degree of node $i$) is is less than $m\sqrt{n} \log(n)$, with probability $1-o(1)$, the ratio between the degree of the $i$'th node in $S$ to the degree of node $i$ is, with probability $1-o(1)$, at most $O(1/\log^3(n))$. The result now follows from the union bound.

\end{proof}

We note that one cannot hope to extend Corollary \ref{thm.highdeg} to general graphs. This is discussed at the end of the proof of Theorem~\ref{lemma.lbcoverage} below. 

\paragraph{Maximizing coverage versus cost.} 
Consider a setting where that accessing a node comes with some fixed cost $c$ one would like to find a set of nodes $S$ such that the effective ``gain per cost" is maximized, where gain per cost is the ratio between the size of $D(S)$ and the total cost of $S$, $c|S|$ (see ~\cite{KhullerMN99}
for an extended variant of the problem).  If $v$ is a node of maximum degree in the graph the solution is to choose such a node $v$.  A potential approximation strategy
would be to quickly find a node of high degree. The following corollary follows from theorem~\ref{thm.highdeg}.
\begin{corollary}
\label{thm.bangforbuck}
Let $G$ be a preferential attachment graph on $n$ nodes. Then, with probability $1-o(1)$, over the probability space generating the preferential attachment graph,
algorithm TraveseToTheRoot, a $1$-local information algorithm, returns a set of size at most $O(\log^4(n))$ 
containing a node of maximum degree in the graph . In particular, TraveseToTheRoot achieves an $O(\log^8(n))$
approximation to the gain per cost coverage problem. 
\end{corollary}
\begin{proof}
With probability $1-o(1)$, after $O(\log^4(n))$ steps, a node of degree at least $\frac{\sqrt{n}}{4\log^3(n)}$ is found,
achieving a gain per cost of $\frac{\sqrt{n}}{5\log^7(n)}$.
As the highest degree is at most $\sqrt{n} \log(n)$, so is the maximum gain per cost and the result follows. 
\end{proof}

Given the positive result for preferential attachment graphs one may ponder
whether it could be extended to work on general graphs. The following theorem would show that in general, algorithms that even $r = n^{\frac{1}{4}}$-local information algorithms, cannot achieve a good approximation. 

\begin{theorem}
\label{lemma.lbcoverage}
Let $k$, $r$, $n \ge \text{max} \{k , r \}$, be positive integers.
For any $r$-local algorithm for the ``gain per cost" problem with success probability bigger than $\frac{1}{3}$, 
the expected approximation over successful runs is $\Omega(\frac{\sqrt{n}}{kr^{2}})$ on k-edge-connected graphs.

In particular, for $k=1$ and $r =o(n^{\frac{1}{4}})$ the approximation ratio grows to infinity with the number of graph nodes $n$. 
\end{theorem}

\begin{proof}
We first focus our attention on proven the claim for $1$-local algorithms on $k=1$ connected graphs. 
The proof will follow similar lines to that of theorem~\ref{lemma.lbconnect}. We invoke the application of Yao minmax principle for the performance of Monte Carlo randomized algorithms \cite{Yao77}.
For that we focus on analyzing the performance of deterministic Monte Carlo algorithms on a uniformly at random chosen input from a family of inputs.
Given $n$ we construct the family of inputs as follows: each input is a graph made from
a complete binary tree on $n-1-\sqrt{n}$ nodes labeled $1,2, \ldots ,n-1-\sqrt{n}$. 
In addition one leaf node $s$ would be a hub for $\sqrt{n}$ new spoke nodes. We denote the subgraph on node $s$ and its neighbors by $H$. Note that node $s$ is the only node with degree bigger than three and so any algorithm that want to achieve a good ``gain per cost" must find that node.

Each input would correspond to a specific choice of assignment for node $i$. In total there are therefore ${(n-1-\sqrt{n})/2 \choose 2}$ inputs.
Such algorithms know only the degrees of the nodes they already traversed. The input comes with a compelling property: if the algorithm has not found a node in the subgraph $H$ after $i$ queries then we learned nothing about the identity of $s$ except that it is one of the leaf nodes not queried so far\footnote{the deterministic algorithm ``knows" the distribution over inputs, i.e.\ that $s$ is a leaf connected to a star subgraph}. Since $s$ was chosen uniformly at random across all 
$\frac{n-1-\sqrt{n}}{2} $ leaves, the probability that after $\frac{\sqrt{n}}{2}$ Jump and Crawl queries a node in $H$ is found is less than $\frac{5}{6}$. To see than we note that with the many Jumps the probability of  
hitting $H$, for $n$ large enough is smaller than $\frac{1}{e}+\frac{1}{100}$. The probability of hitting the leaf $s$ between all leaf trees, given we don't hit the spokes of $H$ is at most $1/2$. By the union bound the total probability of finding $s$ is at most $\frac{1}{2}+\frac{1}{2}(\frac{1}{e}+\frac{1}{100}) < \frac{5}{6}$. Thus,
for any algorithm that is successful with probability, say $1-\frac{2}{6}$, the expected cost of finding node $s$ is
$\Omega(\sqrt{n})$. 

By Yao's principle the expected cost of a randomized algorithm on one of the inputs would be at least $\Omega(\sqrt{n})$. However, an algorithm with full knowledge of the graph can find node $s$ with at most 
$\Theta(\log(n))$ queries on any of the inputs. We conclude that the approximation of any $0$-local algorithm on one of the inputs would be $\Omega(\frac{\sqrt{n}}{\log(n)})$. 

To generalize the problem to $k$ connected graphs 
we replace each edge $(u,v)$ in the complete binary tree subgraph in the construction above by a distinct path
from $u$ to $v$ of length $rk$.  We then connect all the first $k$ nodes on each of the new paths replacing edges  in the original graph to form a clique between themselves, and do the same for any next consecutive blocks of $k$ nodes on that path.

The graph then becomes $k$ edge connected. The total number of nodes becomes $C = n+(n-1)(kr-2)$.
The same compelling property still hold: if the algorithm has not found a node in $H$ after $i$ queries then we learn nothing about the identity of $s$ except it is not one of the node queried so far. 
The expected cost for finding node $s$ would be at least $\Theta(\frac{\sqrt{C}}{rk})$. As an algorithm with full information of the graph can find node $s$ in at most $\log(C)r$ time, the result follows.

We end by noting that as each node has degree at most three except node $s$ the proof also provides a similar lower bound for finding a node 
who is at most a poly-logarithmic factor smaller than the maximum degree, a problem discussed in Section \ref{section.PAapps}.
\end{proof}

\subsection{Analysis of TraverseToTheRoot}
\label{sec.proof.PA}

Our proof will make use of an alternative specification of the preferential attachment process, which is now standard in the literature~\cite{BR-01},~\cite{DFF-11}.  We will now describe this model briefly.  Sample $mn$ pairs $(x_{i,j}, y_{i,j})$ independently and uniformly from $[0,1] \times [0,1]$ with $x_{i,j} < y_{i,j}$ for $i \in [n]$ and $j \in [m]$.  We relabel the variables such that $y_{i,j}$ is increasing in lexicographic order of indices.  We then set $W_0 = 0$ and $W_i = y_{i,m}$ for $i \in [n]$.  We define $w_i = W_i - W_{i-1}$ for all $i \in [n]$.
We then generate our random graph by connecting each node $i$ to $m$ nodes $p_1(i), \dotsc, p_m(i)$, where each $p_k(i)$ is a node chosen randomly with $\Prob[p_k(i) = j] = w_j / W_i$ for all $j \leq i$.  We refer to the nodes $p_k(i)$ as the \emph{parents} of $i$.

Bollob\'{a}s and Riordan showed that the above random graph process is equivalent\footnote{As has been observed elsewhere \cite{DFF-11}, this process differs slightly from the preferential attachment process in that it tends to generate more self-loops.  However, it is easily verified that all proofs in this section continue to hold if the probability of self-loops is reduced.} to the preferential attachment process.  They also show the following useful properties of this alternative model.  Set $s_0 = 160 \log(n)(\log\log(n))^2$ and $s_1 = \frac{n}{2^{25} \log^2 n}$.  Let $I_t = [2^t+1, 2^{t+1}]$.  Define constants $\beta = 1/4$ and $\zeta = 30$.

\begin{lemma}[adapted from \cite{BR-01}]
\label{lem.events}
Let $m \geq 2$ be fixed.  Using the definitions above, each of the following events holds with probability $1 - o(1)$:
\begin{itemize}
\item
$E_1 = \{ |W_i - \sqrt{\frac{i}{n}} | \le \frac{1}{100}\sqrt{\frac{i}{n}}~~\text{for}~s_0 \le i \le n \}$.
\item 
$E_2 = \{ I_t \text{ contains at most } \beta|I_t| \text{ nodes } i~~\text{with}~w_i < \frac{1}{\zeta \sqrt{in}} \text{ for } \log(s_0) \le t \le \log(s_1) \}$.  
\item 
$E_3 = \{ w_1 \geq \frac{4}{\log n \sqrt{n}} \}$.  
\item
$E_4 = \{ w_i \geq \frac{1}{\log^{1.9}(n) \sqrt{n}}~~\text{for all}~i \leq s_0 \}$.  
\item
$E_5 = \{ w_i  \le \frac{\log(n)}{\sqrt{in}}~~\text{for}~s_0 \le i \le n \}$
\end{itemize}
\end{lemma}

Note that we modified these events slightly (from \cite{BR-01}) for our purposes:
event $E_2$ uses different constants $\beta$ and $\zeta$, and in event $E_4$ we provide a bound on $w_i$ for all $i \leq s_0$ rather than $i \leq n^{1/5}$.  Finally, event $E_5$ is a minor variation on the corresponding event from \cite{BR-01}.  The proof of the modified Lemma \ref{lem.events} follows that of Bollob\'{a}s and Riordan \cite{BR-01} quite closely and can be found at the end of the chapter.

Given Lemma \ref{lem.events}, we can think of the $W_i$'s as arbitrary fixed values that satisfy events $E_1, \dotsc, E_5$, rather than as random variables.  Lemma \ref{lem.events} implies that, if we can prove Theorem \ref{thm.PA} for random graphs corresponding to all such sequences of $W_i$'s, then it will also hold for preferential attachment graphs. 
We now turn to the proof of Theorem \ref{thm.PA}.  Let us provide some intuition.  We would like to show that TraverseToTheRoot queries nodes of progressively higher degrees over time. However, if we query a node $i$ of degree $d$, there is no guarantee that subsequent nodes will have degree greater than $d$; the algorithm may encounter local maxima.  
Suppose, however, that there were a path from $i$ to the root consisting entirely of nodes with degree at least $d$.  In this case, the algorithm will only ever traverse nodes of degree at least $d$ from that point onward.  One might therefore hope that the algorithm finds nodes that lie on such ``good'' paths for ever higher values of $d$, representing progress toward the root.

Motivated by this intuition, we will study the probability that any given node $i$ lies on a path to the root consisting of only high-degree nodes (i.e.\ not much less than the degree of $i$).  We will argue that many nodes in the network lie on such paths.  We prove this in two steps.  First, we show that for any given node $i$ and parent $p_k(i)$, $p_k(i)$ will have high degree relative to $i$ with probability greater than $1/2$ (Lemma \ref{lem.lem1}).  Second, since each node $i$ has at least two parents, we use the theory of supercritical branching processes to argue that, with constant probability for each node $i$, there exists a path to a node close to the root following links to such ``good'' parents (Lemma \ref{lem.lem2}).

This approach is complicated by the fact that existence of such good paths is highly correlated between nodes; this makes it difficult to argue that such paths occur ``often'' in the network.  To address this issue, we show that good paths are likely to exist even after a large set of nodes ($\Gamma$ in our argument below) is adversarially removed from the network.  We can then argue that each node is likely to have a good path independently of many others nodes, as we can remove all nodes from one good path before testing the presence of another.

We will now proceed with the details of the proof.  The proofs of all technical lemmas appear at the end of this section. 
Set $s_0 = 160 \log(n)(\log\log(n))^2$ and $s_1 = \frac{n}{2^{25} \log^2 n}$.  We think of vertices in $[1,s_0]$ as close to the root, and vertices in $[s_1, n]$ as very far from the root.  Let $I_t = [2^t+1, 2^{t+1}]$ be a partition of $[n]$ into intervals.  Define constants $\beta = 1/4$ and $\zeta = 30$.
We now define what we mean by a typical node.
\begin{definition}[Typical node]
A node $i$ is \emph{typical} if either $w_i \geq \frac{1}{\zeta \sqrt{in}}$ or $i \le s_0$.
\end{definition}

Note that event $E_2$ implies that each interval $I_t, ~\log(s_0) \le t \le \log(s_1)$ contains a large number of typical nodes.  

The lemma below encapsulates concentration bounds on the degrees of nodes in the network as well as other useful  properties of PA networks.

\begin{lemma}\label{lemma.degree}
The following events hold with probability $1-o(1)$:
\begin{itemize}
\item
$E_6 = \{ \forall i \ge s_0:~  \deg(i) \le 6m\log(n) \sqrt{\frac{n}{i}} \}$.
\item
$E_7 = \{ \forall s_0 \le i \le s_1 \text{ that is typical}:~  \deg(i) \ge \frac{m}{2\zeta}\sqrt{\frac{n}{i}} \}$.
\item
$E_8 = \{ \forall i \le s_0:~  \deg(i) \ge \frac{m\sqrt{n}}{5\log^{1.9}(n)} \}$.
\item
$\forall i \ge s_0:~  \Prob[i \text{ is connected to } 1] \ge \frac{3.9}{\log(n)\sqrt{i}}$.
\item
$\forall j \ge i \ge s_0, 1 \leq k \leq m:~  \Prob[p_k(j) \leq i] \ge \frac{0.9\sqrt{i}}{\sqrt{j}}$.
\end{itemize}
\end{lemma}

Our next lemma states that, for any set $\Gamma$ that contains sufficiently few nodes from each interval $I_t$, and any given parent of a node $i$, with probability greater than $1/2$ the parent will be typical, not in $\Gamma$, and not in the same interval as $i$.  

\begin{definition}[Sparse set] 
A subset of nodes $\Gamma \subseteq [n]$ is \emph{sparse} if $|\Gamma \cap I_t| \leq |I_t|/\log\log(n)$ 
for all $\log s_0 \leq t \leq \log s_1$.  That is, $\Gamma$ does not contain more than a $1/\log\log n$ fraction of the nodes in any interval $I_t$ contained in $[s_0, s_1]$.
\end{definition}

\begin{lemma}
\label{lem.lem1}
Fix sparse set $\Gamma$.  Then for each $i \in [s_0, s_1]$ and $k \in [m]$, the following are true with probability $\geq 8/15$ : $p_k(i) \not\in \Gamma$, $p_k(i) \leq i / 2$, and $p_k(i)$ is typical.
\end{lemma}

We now claim that, for any given node $i$ and sparse set $\Gamma$, there is likely a short path from $i$ to vertex $1$ consisting entirely of typical nodes that do not lie in $\Gamma$.  Our argument is via a coupling with a supercritical branching process.  Consider growing a subtree, starting at node $i$, by adding to the subtree any parent of $i$ that satisfies the conditions of Lemma \ref{lem.lem1}, and then recursively growing the tree in the same way from any parents that were added.  Since each node has $m \geq 2$ parents, and each satisfies the conditions of Lemma \ref{lem.lem1} with probability $> 1/2$, this growth process is supercritical and should survive with constant probability (within the range of nodes $[s_0, s_1]$).  We should therefore expect that, with constant probability, such a subtree would contain some node $j < s_0$.  

To make this intuition precise we must define the subtree structure formally.  Fix sparse set $\Gamma$ and a node $i \in [s_0, s_1]$. 
Define $H_\Gamma(i)$ to be the union of a sequence of sets $H_0, H_1, \dotsc$, as follows.  First, $H_0 = \{i\}$.  Then, for each $\ell \geq 1$, $H_\ell$ will be a subset of all the parents of the nodes in $H_{\ell-1}$.  For each $j \in H_{\ell-1}$ and $k \in [m]$, we will add $p_k(j)$ to $H_\ell$ if and only if the following conditions hold:
\begin{enumerate}
\item $p_k(j)$ is typical, $p_k(j) \not\in \Gamma$, and $p_k(j) \leq j / 2$,
\item $p_k(j) \not\in H_r$ for all $r \leq \ell$, and
\item For the interval $I_t$ containing $p_k(j)$, $|I_t \cap (H_0 \cup \dotsc \cup H_\ell)| < 10\log\log n$.  
\end{enumerate}

Item 1 contains the conditions of Lemma \ref{lem.lem1}.  Item 2 is that $p_k(j)$ has not already been added to the subtree; we add this condition so that the set of parents of any two nodes in the subtree are independent.  Item 3 is that the subtree contains at most $10\log\log n$ nodes from each $I_t$.  We will use this condition to argue that $\Gamma$ remains sparse if we add all the elements of $H_\Gamma(i)$ to $\Gamma$.

Our next lemma states that any given node $i \in [s_0,s_1]$ has a short path to the root consisting of only typical nodes with probability at least $3/4$.

\begin{lemma}
\label{lem.lem2}
Fix any sparse set $\Gamma$.  Then for each node $i \in [s_0, s_1]$, the probability that $H_{\Gamma}(i)$ contains a node $j \leq s_0$ is at least $1/5$.
\end{lemma}

Lemma \ref{lem.lem2} implies the following result, which we will use in our analysis of the algorithm TraverseToTheRoot.  First a definition.  
\begin{definition}[Good path]
A path $(j_0, j_1, ..., j_k)$ is \emph{good} if $j_k \leq s_0$, each $j_\ell$ is typical and, for each $\ell > 0$, $ j_\ell \leq j_{\ell-1}/2$.  We say vertex $i \in [s_0,s_1]$ \emph{has a good path} if there is a good path with $j_0 = i$.

\end{definition}

\begin{lemma}
\label{lem.good.paths}
Choose any set $T$ of at most $16 \log n$ nodes from $[s_0, s_1]$.  Then each $i \in T$ has a good path with probability at least $1/5$, independently for each $i$.
\end{lemma}

We will apply Lemma \ref{lem.good.paths} to the set of nodes queried by TraverseToTheRoot to argue that progress toward the root is made after every sequence of polylogarithmically many steps.

We can now complete the proof of Theorem \ref{thm.PA}, which we give below.

\begin{proof}[Proof of Theorem \ref{thm.PA}]
Our analysis will consist of three steps, in which we consider three phases of the algorithm.  The first phase consists of all steps up until the first time TraverseToTheRoot traverses a node $i < s_1$ with a good path.
The second phase then lasts until the first time the algorithm queries a node $i < s_0$.  Finally, the third phase ends when the algorithm traverses node $1$. We will show that each of these phases lasts at most $O(\log^4(n))$ steps.  

We note that we will make use of Lemma \ref{lem.good.paths} in our analysis by way of considering whether certain nodes have good paths. We will check at most $16\log n$ nodes in this manner, and hence the conditions of Lemma \ref{lem.good.paths} will be satisfied.

\paragraph{Analysis of phase 1}
Phase $1$ begins with the initial node $u$, and ends when the algorithm traverses a node $i < s_1$ with a good path.  
We divide phase 1 into a number of iterations.  Iteration zero starts at node $u$.  Define iteration $t$ as the first time, after iteration $t-1$, that the algorithm queries a node $i \le s_1$.

Each new node $i$ considered in iteration $t$ will have $i < s_1$ with probability at least $W_{s_1}/1 \geq \frac{1}{2^{13} \log n}$, regardless of the previous nodes traversed. By the multiplicative Chernoff bound (\ref{lemma.multchernoff}), with probability of at 
least $1-1/n^{2}$, after at most $5\log^2(n)$ steps such a node $i <s_1$ would be found. 
By Lemma \ref{lem.good.paths} we know that node $i$ has a good path with probability at least $1/5$ independent of all nodes traversed so far. 

By the multiplicative Chernoff bound (\ref{lemma.multchernoff}), we conclude that after
at most $10\log(n)$ iterations, and total time of $O(\log^3(n))$, the algorithm traverses a node that has both 
$i < s_1$ and a good path, with probability at least $1-\frac{\log(n)}{n^2}-\frac{1}{n}$. 

We note that the number of invocations of Lemma \ref{lem.good.paths} made during the analysis of this phase is at most $2 \log n$ with high probability, and hence the cardinality restriction of Lemma \ref{lem.good.paths} is satisfied.

\paragraph{Analysis of phase 2}
Phase 2 begins once the algorithm has traversed some node $i < s_1$ with a good path, and ends when the algorithm traverses a node $j < s_0$.  
We split phase $2$ into a number of epochs.  For each $\log s_0 < t \leq \log s_1$, we define epoch $t$ to consist of all steps of the algorithm during which some node $i \in I_t$ with a good path has been traversed, but no node in any $I_\ell$ for $\ell < t$ with a good path has been traversed.  Define random variable $Y_t$ to be the length of epoch $t$.  Note phase 2 ends precisely when epoch $\log s_0$ ends.  Further, the total number of steps in phase 2 is $\sum_{t = \log s_0}^{\log s_1} Y_t$.

Fix some $\log s_0 \leq t \leq \log s_1$ and consider $Y_t$.  Suppose the algorithm is in epoch $t$, and let $i \in I_t$ be the node with a good path that has been traversed by the algorithm.  Then, from the definition of a good path and event $E_7$, $i$ has a parent $j \in I_\ell$ for some $\ell < t$ with $deg(j) \geq \frac{m}{2\zeta}\sqrt{\frac{n}{i}}$.  This node $j$ is a valid choice to be traversed by the algorithm, so any node queried before $j$ must have degree at least $\frac{m}{2\zeta}\sqrt{\frac{n}{i}}$.  Moreover, traversing node $j$ would end epoch $t$, so every step in epoch $t$ traverses a node with degree at least $\frac{m}{2\zeta}\sqrt{\frac{n}{i}}$.  By event $E_6$, any such node $\ell$ satisfies $\ell < z i \log^2(n)$ for constant $z = (4\zeta)^2$.  But we now note that, for any node $\ell < z i \log^2(n)$ traversed by the algorithm, the probability that $\ell$ has a parent\footnote{Note that even if the algorithm queried node $\ell$ via its connection to one of its parents, it will still have at least one other parent that is independent of prior nodes queried by the algorithm since $m \geq 2$.} 
$r < i/2$ is at least $\frac{W_{i/\log^2(n)}}{W_\ell} \geq \frac{1}{4\zeta \log^2 n}$.  Any such node $r$ has degree greater than any node in $I_t$, again by Lemma \ref{lemma.degree}, so if a queried node had such a parent then the subsequent step must query a node of index at most $2^t$.  Moreover, Lemma \ref{lem.good.paths} implies that this node of index at most $2^t$ has a good path with probability at least $1/5$.  Thus each step of the algorithm results in the end of epoch $t$ with probability at least $\frac{1}{20\zeta \log^2 n}$.  We conclude that $Y_t$ is stochastically dominated by a geometric random variable with mean $20\zeta \log^2 n$.  Also, the number of invocations of Lemma \ref{lem.good.paths} made during epoch $t$ is dominated by a geometric random variable with mean $5$.

We conclude that $\sum_{t = \log s_0}^{\log s_1} Y_t$ is dominated by the sum of at most $\log n$ geometric random variables, each with mean $20\zeta \log^2 n = 600 \log^2 n$.  Concentration bounds for geometric random variables (Lemma \ref{lemma.geomchernoff}) now imply that, with high probability, this sum is at most $2^{10} \log^3 n$.  We conclude that phase $2$ ends after at most $2^{10} \log^3 n$ steps with high probability.  Similarly, the total number of invocations of Lemma \ref{lem.good.paths} made during the analysis of this phase is at most $6\log n$ with high probability, again by Lemma \ref{lemma.geomchernoff}.

\paragraph{Analysis of phase 3}
We turn to analyze the time it takes from the first time the algorithm encountered a node of $i \le s_0$ until node $1$ is found.
We start by noting that the induced graph on the first $s_0$ nodes is connected with probability $1-o(1)$ (see for example corollary $5.15$ in \cite{DFF-11}, used with $n := log(n)$). We note that by Lemma~\ref{lemma.degree} every node $j \le s_0$ has degree at least 
$d = \frac{m\sqrt{n}}{5\log^{1.9}(n)}$.
As there is a path from $i$ to node $1$ where all nodes have degree at least $d$,
the algorithm, as it follows the highest neighbor of its current set $S$, will reach node $1$ before it had traversed any node of degree less than $d$. We can therefore assume that the algorithm only traverses nodes of degree greater than $d$.

By Lemma ~\ref{lemma.degree}, each node $j > s_0$ has $deg(j) \leq 6m \log(j)\sqrt{\frac{n}{j}}$ with high probability, and therefore any node $j$ with degree $> d$ must satisfy $j < (60\zeta)^2 \log^{5.8}(n)$.  For any such node, $E_1$ implies that $W_j \leq \frac{11}{10}\frac{(60\zeta) \log^{2.9}(n)}{\sqrt{n}}$.  Thus, for each such $j$, the probability that $j$ is connected to the root is $w_1 / W_j \geq \frac{1}{2^{11} \log^{3.9}(n)}$, by event $E_3$.  
Chernoff bounds (Lemma \ref{lemma.multchernoff}) then imply that such an event will occur with high probability after at most $O(\log^4(n))$ steps, so, with high probability, phase $3$ will end after at most $s_0 + O(\log^4(n)) = O(\log^4(n))$ steps. 

\end{proof}

\section{Omitted Proofs from Section \ref{sec.proof.PA}}
\label{sec.PA.lemmas}

\subsection{Proof of Lemma \ref{lem.events}}
We provide details for the proof of Lemma \ref{lem.events}.  This result follows that of Bollob\'{a}s and Riordan \cite{BR-01} quite closely; we present the differences only briefly for completeness.  

The proof that events $E_1$, $E_2$, and $E_3$ hold with high probability follows entirely without change, except for the modification of certain constants.  We therefore omit the details here.

We next show that event $E_4$ holds with high probability, by showing that $\Pr[E_4^c \cap E_1] = o(1)$.  Suppose that $E_4^c \cap E_1$ holds and let $\delta = \frac{1}{\log^{1.9}(n)\sqrt{n}}$.  As $E_1$ holds we have $W_{s_0} \leq \frac{11 \log\log(n)\sqrt{\log(n)}}{10 \sqrt{n}}$.  As $E_4$ does not hold there exists some interval $[x,x+\delta]$ with $0 \leq x \leq \frac{11 \log\log(n)\sqrt{\log(n)}}{10 \sqrt{n}}$ that contains two of the $W_i$ and hence two of the $y_{i,j}$.  Each such interval is contained in some interval $J_t = [t\delta, (t+2)\delta]$ for $0 \leq t \leq \delta^{-1}\frac{11 \log\log(n)\sqrt{\log(n)}}{10 \sqrt{n}} < 2 \log^{2.5}(n)$.  The probability that some $y_{i,j}$ lands in such an interval is $(4t+4)\delta^2$, so the probability that at least two lie in $J_t$ is at most $m^2n^2(4t+4)^2\delta^4/2 < 32m^2/\log^{2.6}(n)$.  Thus
\[ \Prob(E_4^c \cap E_1) \leq \sum_{t=0}^{2 \log^{2.5}(n)}32m^2/\log^{2.6}(n) = o(1) \]
as required.

We will next show that the event $E_5$ holds with high probability.  Recall that event $E_5$ is $\{ w_i  \le \frac{\log(n)}{\sqrt{in}} \text{ for } s_0 \le i \le n \}$.  We will show that $\Prob(E_5^c \cap E_1) = o(1)$, which will imply that $\Prob(E_5) = 1 - o(1)$ as required.

Suppose that $E_5^c \cap E_1$ holds.  Then there is some $i \geq s_0$ is such that $w_i > \frac{\log(n)}{\sqrt{in}}$.  Define $\delta = \frac{\log(n)}{\sqrt{in}}$; it must therefore be that the interval $(W_{i-1}, W_{i-1}+\delta]$ does not contain $W_i$, and hence contains at most $m-1$ of the $y_{i,j}$.  Since $E_1$ holds, we must have $W_{s_0} \geq \frac{9}{10}\sqrt{\frac{s_0}{n}}$.  We now define a partition of $[\frac{9}{10}\sqrt{\frac{s_0}{n}}, 1]$ into intervals $J_t = [x_t, x_{t+1})$ for $t \geq 0$, where we define $x_{0} = \frac{9}{10}\sqrt{\frac{s_0}{n}}$ and $x_t = x_{t-1} + \frac{\log(n)}{x_{t-1}nm}$ for all $t \geq 1$, until $x_t \geq 1$.  We note that there are at no more than $mn$ intervals $J_t$ in total.  We also note that, since $E_1$ holds, each interval $(W_{i-1}, W_{i-1}+\delta]$ contains at least $m-1$ intervals $J_t$, each satisfying $x_t \geq W_{i-1}$, one of which must contain no $y_{i,j}$ since $E_5$ does not hold.

For a given $t$ satisfying $x_t \geq W_{i-1}$, the number of $y_{i,j}$ in $J_t$ has a $Bi(mn,p_t)$ distribution with
\[ p_t = x_{t+1}^2 - x_t^2 \leq 2x_t \frac{\log(n)}{x_{t}nm} = \frac{2\log(n)}{nm}. \]
The probability that no $y_{i,j}$ lies in this interval is thus
\[ (1-p_t)^{mn} \leq e^{-mnp_t} < e^{-2\log(n)} = o(n^{-1}). \]
Summing over the $O(n)$ values of $t$ shows that $\Pr(E_5^c \cap E_1) = o(1)$, as required.

\subsection{Proof of Lemma \ref{lemma.degree}}

We will first prove that the following events hold with probability $1 - o(1)$:

\begin{itemize}
\item
$E_6 = \{ \forall i \ge s_0:~  \deg(i) \le 6m\log(n) \sqrt{\frac{n}{i}} \}$.
\item
$E_7 = \{ \forall s_0 \le i \le s_1 \text{ that is typical}:~  \deg(i) \ge \frac{m}{2\zeta}\sqrt{\frac{n}{i}} \}$.
\item
$E_8 = \{ \forall i \le s_0:~  \deg(i) \ge \frac{m\sqrt{n}}{5\log^{1.9}(n)} \}$.
\end{itemize}

Note that event $E_7$ states that typical nodes have typical degree, motivating our choice of terminology.

We start by noting that $\deg(i)= \sum_{j=i+1}^{n}\sum_{k=1}^{m}{Y_{k,j}}$
where each of the $Y_{i,j}$s is an i.i.d Bernoulli random variable that gets the value of one with success probability of $\frac{w_i}{W_j}$. This follows from the fact the each new node $j$ sends m edges backwards and the probability of each hitting node $i$ is exactly $\frac{w_i}{W_j}$. 
From $E_1$ and $E_5$, $$\E(\deg(i)) \le \sum_{j=i+1}^{n} \left( m\frac{\log(n)\frac{1}{\sqrt{in}}}{\frac{9}{10}\sqrt{\frac{j}{n}}} \right) = 
\sum_{j=i+1}^{n} \left( m\log(n)\frac{10}{9} \frac{1}{\sqrt{ij}} \right).$$ By estimating the sum with an integral we get
$$\E(\deg(i)) \le \frac{10m}{9}\log(n) \frac{\sqrt{n}}{\sqrt{i}} .$$ From the multiplicative Chernoff bound (\ref{lemma.multchernoff}) 
we conclude that with probability bigger than $1-1/n^2$, $\deg(i) \le 3m\log(n) \frac{\sqrt{n}}{\sqrt{i}}$ for a given node $i$. By using the union bound, event $E_6$ then holds with probability $1-1/n$.

To prove $E_7$ holds with probability $1-1/n$, 
we first recall that, for a typical node $i$, 
$w_i \ge \frac{m}{\zeta\sqrt{in}}$. 
This implies, similarly to the first part of the proof, 
that $$\E(\deg(i)) \ge \frac{10m}{11} \frac{\sqrt{n}}{\zeta \sqrt{i}} .$$ As $\text{Exp}(\deg(i))  \ge 16m \log(n)$
for any $ s_0 \le i \le s_1$ (since $s_1 = \frac{n}{2^{25}\log^2 n}$), we can invoke the Chernoff bound (\ref{lemma.multchernoff}) to get that $E_7$ holds with probability bigger than $1-1/n^2$ for a given node $i$. This follows by thinking of $\deg(i)$ as a sum of Bernoulli random variables $Y_{i,j}$,
where $Y_{i,j}$ succeeds with probability $\frac{\frac{1}{\zeta\sqrt{in}}}{W_j}$.
By using the union bound, event $E_7$ then holds with probability $1-1/n$.

The proof that $E_8$ holds with probability $1-1/n$ follows similarly to the proof for such a claim for $E_7$,
 by using the property that $w_i \ge \frac{1}{\log^{1.9}(n) \sqrt{n}}$.

To complete the proof of Lemma \ref{lemma.degree}, we must show that
\begin{itemize}
\item
$\forall i \ge s_0:~  \Prob[i \text{ is connected to } 1] \ge \frac{3.9}{\log(n)\sqrt{i}}$, and
\item
$\forall j \ge i \ge s_0, 1 \leq k \leq m:~  \Prob[p_k(j) \leq i] \ge \frac{0.9\sqrt{i}}{\sqrt{j}}$.
\end{itemize}
The first item follows from events $E_1$ and $E_3$ of Lemma \ref{lem.events}, plus the fact that $\Prob[p_k(i) = 1] = \frac{w_1}{W_i}$ for every $i$ and $k$.  The second item follows from event $E_1$ of Lemma \ref{lem.events}, plus the fact that $\Prob[p_k(j) \leq i] = \frac{W_i}{W_j}$.

\subsection{Proof of Lemma \ref{lem.lem1}}

We first recall the statement of the lemma.  Fix any sparse set $\Gamma$.  Then for each $i$, $s_0 \leq i \leq s_1$, and each $k \in [m]$, the following statements are all true with probability at least $8/15$ : $p_k(i) \not\in \Gamma$, $p_k(i) \leq i / 2$, and $p_k(i)$ is typical.

Fix $i$ and $k$.  For each of the three statements in the lemma, we will bound the probability of that statement being false.

First, we will show that $\Prob[p_k(i) \text{ not typical }] < 1/15$.  Note that, given that $p_k(i)$ falls within an interval $I_t$, this probability is bounded by the total weight of atypical nodes in $I_t$ divided by the total weight of $I_t$.  Since each atypical node $j$ has weight at most $\frac{1}{10 \sqrt{j n}}$ and $j > 2^t$ for all $j \in I_t$, $E_4$ implies that the total weight of atypical nodes in $I_t$ is at most 
\[ \beta|I_t|\frac{1}{10 \sqrt{2^t n}} = \beta\frac{\sqrt{2^t}}{10\sqrt{n}}. \]
Also, $E_1$ implies that the total weight of $I_t$ is 
\[ W_{2^{t+1}} - W_{2^t} \leq \sqrt{\frac{2^t}{n}}\left(\frac{99}{100}\sqrt{2} - \frac{101}{100}\right). \]
Since these bounds hold for all $t$, we conclude that 
\[ \Prob[p_k(i) \text{ not typical }] < \frac{10\beta}{99\sqrt{2} - 101} \]
which will be at most $1/15$ for $\beta = 1/4$.

Next, we will show that $\Prob[p_k(i) > i/2] < \frac{1}{3}$.  Event $E_1$ implies that
\[ \Prob[p_k(i) > i/2] = 1 - W_{i/2} / W_i \leq 1 - \frac{99}{101\sqrt{2}} < \frac{1}{3}. \] 

Finally, we will show that $\Prob[p_k(i) \in \Gamma] < 1/15$.  Given that $p_k(i)$ falls within an interval $I_t$, this probability is bounded by the total weight of $I_t \cap \Gamma$ divided by the total weight of $I_t$.  In this case, due to the assumed sparsity of $\Gamma$ and $E_5$, the former quantity is at most 
$|I_t|\frac{1}{(\log\log n)\sqrt{2^t n}} \leq \sqrt{\frac{2^t}{n}}.$
Also, as above, the total weight of $I_t$ is at most $\sqrt{\frac{2^t}{n}}(\frac{99}{10}\sqrt{2} - \frac{101}{10})$. 
Since these bounds hold for all $t$, we conclude that $\Prob[p_k(i) \in \Gamma] < \frac{1}{99\sqrt{2} - 101}$
which is at most $1/15$.

Taking the union bound over these three events, we have that the probability none of them occur is at least $8/15$ as required.

\subsection{Proof of Lemma \ref{lem.lem2}}

Let us first recall the statement of the lemma.  Fix any sparse set $\Gamma$.  Then for each node $i \in [s_0, s_1]$, the probability that $H_{\Gamma}(i)$ contains a node $j \leq s_0$ is at least $1/5$.

Fix $\Gamma$ and $i$, and write $H = H_{\Gamma}(i)$.  Let $C = [s_0]$, the set of all nodes with index $s_0$ or less.  We will show that the probability that $H \cap C = \emptyset$ is at most $4/5$.

Let $\ell$ be such that $i \in I_\ell$.  We will say that $H$ \emph{saturates} a given interval $I_t$ if $|H \cap I_t| = 10\log\log n$.  (Note that we must have $|H \cap I_t| \leq 10\log\log n$, from the definition of $H$).  Let us first consider the probability that $H \cap C = \emptyset$ and $H$ does not saturate any intervals.  Since $H$ does not saturate any intervals, and since the set $H \cup \Gamma$ is itself a sparse set, then for each node $j \in H$ and $k \in [m]$ the parent $p_k(j)$ will be added to $H$ precisely if the conditions of Lemma \ref{lem.lem1} hold, which occurs with probability at least $8/15$.  We can therefore couple the growth of the subtree $H$ within the range $[s_0, i]$ with the growth of a branching process in which each node spawns up to two children, each with probability at least $8/15$.  In this coupling, the event $H \cap C = \emptyset$ implies the event that this branching process generates only finitely many nodes.  Write $p$ for the probability that the branching process generates infinitely many nodes.  Then $p = \frac{8}{15}p + (1-\frac{8}{15}p)\frac{8}{15}p$, from which we obtain $p = \frac{15}{64}$.  We therefore have $\Prob[H \cap C = \emptyset] \leq 1-p = \frac{49}{64}$ conditional on $H$ not saturating any intervals.
Next consider the probability that $H \cap C = \emptyset$ given that $H$ does saturate some interval.  In this case, there is some smallest $t$ such that $I_t$ is saturated by $H$.  Then, given that $H$ saturates $I_t$ but no interval $I_{t'}$ for $t' < t$, then we can again couple the growth of subtree $H$ from interval $I_t$ onward with $10\log\log n$ instances of the branching process described above, each one starting at a different node in $H \cap I_t$.  In this case,  the probability that $H \cap C = \emptyset$ is bounded by the probability that each of these $10\log\log n$ copies of the branching process all generate only finitely many children.  This probability is at most $(49/64)^{10\log\log n} = o(\frac{1}{\log^2(n)})$.  Thus, taking the union bound over all possibilities for the value of $t$ (of which there are at most $\log n$), the probability that $H \cap C = \emptyset$ given that $H$ saturates some interval is at most $o(\log n / \log^2(n) ) = o(1)$.

Combining these two cases, we see that $\Prob[H \cap C = \emptyset] \leq 49/64 + o(1) < 4/5$.

\subsection{Proof of Lemma \ref{lem.good.paths}}

Write $T = \{t_1, \dotsc, t_k\}$.  We will apply Lemma \ref{lem.lem2} to each node $t_i$ in sequence.  First, for node $t_1$, define $\Gamma_1 = \emptyset$.  Lemma \ref{lem.lem2} with $\Gamma = \Gamma_1$ implies that $H_{\Gamma_1}(t_1)$ contains a node $j \leq s_0$ with probability at least 1/5.

For each subsequent node $t_i$, define $\Gamma_i = \Gamma_{i-1} \cup H_{\Gamma_{i-1}}(i-1)$.  We claim that this $\Gamma_i$ is sparse.  To see this, recall that each $H_{\Gamma}(t_{i-1})$ contains at most $10\log\log n$ nodes in each interval $I_t$, and $\Gamma_i$ is the union of at most $16 \log n$ such sets, so $|\Gamma_i \cap I_t| \leq 160 \log(n)\log\log(n)$ for each $t$.  Since $|I_t| \geq s_0 \geq 160 \log(n)(\log\log(n))^2$, we have that $|\Gamma_i \cap I_t| \leq |I_t|/\log\log(n)$ and hence $\Gamma_i$ is sparse.  Lemma \ref{lem.lem2} with $\Gamma = \Gamma_i$ then implies that $H_{\Gamma_i}(t_i)$ contains a node $j \leq s_0$ with probability at least 1/5.  Moreover, this probability is independent of the events for nodes $t_1, \dotsc, t_{i-1}$, since $H_{\Gamma_i}(t_i)$ is constrained to not depend on nodes in $\Gamma_i$, which contains all nodes that influenced the outcome for $t_1, \dotsc, t_{i-1}$.

We conclude that, for each $i$, $H_{\Gamma_i}(t_i)$ contains a node $j \leq s_0$ with probability at least 1/5, independently for each $t_i$.  For any given $i$, in the case that this event occurs and by the definition of $H_{\Gamma_i}(t_i)$, $H_{\Gamma_i}(t_i)$ contains a path $P$ from $t_i$ to $j$ consisting entirely of typical nodes, all of which are at most $t_i$, and each node on the path $P$ has creation time (index) at most half of that of its immediate predecessor.

%% file: LocalInfo-mindom.tex
\section{Minimum Dominating Set on Arbitrary Networks}
\label{sec.mindom}

We now consider the problem of finding a dominating set $S$ of minimal size for an arbitrary graph $G$.  Even with full (non-local) access to the network structure, it is known to be hard to approximate the Minimum Dominating Set Problem to within a factor of $H(\Delta)$ in polynomial time, via a reduction from the set cover problem, 
where $H(n) \approx \ln(n)$ is the $n$th harmonic number.  In this section we explore how much local network structure must be made visible in order for it to be possible to match this lower bound.  

Guha and Khuller \cite{GK-98} design an $O(H(\Delta))$-approximate algorithm for the minimum dominating set problem, which can be interpreted in our framework as a $2^+$-local algorithm.  
Their algorithm repeatedly selects a node that greedily maximizes the number of dominated nodes, considering only nodes within distance $2$ of a previously selected node.  
As we show, the ability to observe network structure up to distance 2 is unnecessary if we allow the use of randomness: we will construct a randomized $O(H(\Delta))$ approximation algorithm that is $1^+$-local. We then show that this level of local information is crucial: no algorithm with less local information can return a non-trivial approximation.

\subsection{A $1^+$-local Algorithm}

We now present a $1^+$-local randomized $O(H(\Delta))$-approximation algorithm for the min dominating set problem.  Our algorithm obtains this approximation factor both in expectation and with high probability in the optimal solution size\footnote{Our algorithm actually generates a connected dominating set, so it can also be seen as an $O(H(\Delta))$ approximation to the connected dominating set problem.}.  

Roughly speaking, our approach is to greedily grow a subtree of the network, repeatedly adding vertices that maximize the number of dominated nodes.  Such a greedy algorithm is $1^+$-local, as this is the amount of visibility required to determine how much a given node will add to the number of dominated vertices.  Unfortunately, this greedy approach does not yield a good approximation; it is possible for the algorithm to waste significant effort covering a large set of nodes that are all connected to a single vertex just beyond the algorithm's visibility.  To address this issue, we introduce randomness into the algorithm: after each greedy addition of a node $x$, we will also query a random neighbor of $x$.  The algorithm is listed above as Algorithm \ref{alg.mindomset} (AlternateRandom).

\begin{algorithm}
\caption{AlternateRandom} \label{alg.mindomset}
\begin{algorithmic}[1]
\STATE Select an arbitrary node $u$ from the graph and initialize $S = \{u\}$.
\WHILE{$D(S) \neq V$}
         \STATE Choose $x \in \arg\max_{v \in N(S)}\{ |N(v) \backslash D(S)| \}$ and add $x$ to $S$.
         \IF{$N(x) \backslash S \neq \emptyset$}
                   \STATE Choose $y \in N(x) \backslash S$ uniformly at random and add $y$ to $S$.
         \ENDIF
\ENDWHILE
\STATE return $S$.
\end{algorithmic}
\end{algorithm} 

We now show that
AlternateRandom obtains an $O(H(\Delta))$ approximation, both in expectation and with high probability.
In what follows, $\OPT$ will denote the size of the optimal dominating set in an inplicit input graph.

\begin{theorem}
\label{thm.mindom}
AlternateRandom is $1^+$-local and returns a dominating set $S$ where $\EE[|S|] \leq 2(1+H(\Delta)) \OPT + 1$ and $\Prob[|S| > 2(2+H(\Delta)) \OPT] < e^{-\OPT}$.
\end{theorem}

\begin{proof}
Correctness follows from line 2 of the algorithm.  To show that it is $1^+$-local, it is enough to show that line 3 can be implemented by a $1^+$-local algorithm.  This follows because, for any $v \in N(S)$, $|N(v)\backslash D(S)|$ is precisely equal to the degree of $v$ minus the number of edges between $v$ and other nodes in $D(S)$.

We will bound the expected size of $S$ via the following charging scheme.  Whenever a node $x$ is added to $S$ on line 4, we place a charge of $1/|N(x) \backslash D(S)|$ on each node in $N(x) \backslash D(S)$.  
These charges sum to $1$, so sum of all charges increases by $1$ on each invocation of line 4.  
We will show that the total charge placed during the execution of the algorithm is at most $(1 + H(\Delta))\OPT$ in expectation.  This will imply that $\E[(|S|-1)/2] \leq (1+H(\Delta))\OPT$ as required.

Let $T$ be an optimal dominating set.  Partition the nodes of $G$ as follows: for each $i \in T$, choose a set $S_i \subseteq D(\{i\})$ containing $i$ such that the sets $S_i$ form a partition of $G$.  Choose some $i \in T$ and consider the set $S_i$.  We denote by a ``step'' any execution of line $4$ in which charge is placed on a node in $S_i$.  
We divide these steps into two phases: phase $1$ consists of steps that occur while $S_i \cap S = \emptyset$, and phase $2$ is all other steps.  Note that since we never remove nodes from $S$, phase 1 occurs completely before phase 2.

We first bound the total charge placed on nodes in $S_i$ in phase $1$.  In each step, some number $k$ of nodes from $S_i$ are each given some charge $1/z$.  This occurs when $|N(x) \backslash D(S)| = z$ and $(N(x) \backslash D(S)) \cap S_i = k$.  In this case, if phase $1$ has not ended as a result of this step, there is a $k/z$ probability that a node in $S_i$ is selected on the subsequent line 6 of the algorithm, which would end phase $1$.  We conclude that if the total charge added to nodes in $S_i$ on some step is $p \in [0,1]$, phase $1$ ends for set $S_i$ with probability at least $p$.  The following probabilistic lemma now implies that the expected sum of charges in phase $1$ is at most $1$.

\begin{lemma}
\label{lemma.mindom.sequence}
For $1 \leq i \leq n$, let $X_i$ be a Bernoulli random variable with expected value $p_i \in [0,1]$.  Let $T$ be the random variable denoting the smallest $i$ such that $X_i = 1$ (or $n$ if $X_i = 0$ for all $i$).  Then $\E_T\left[\sum_{i=1}^T p_i \right] \leq 1$.
\end{lemma}
\begin{proof}
We proceed by induction on $n$.  The case $n = 1$ is trivial.  For $n > 1$, we note that 
\[ \E_T\left[\sum_{i=1}^T p_i \right] = p_1 + (1-p_1) \E_T\left[\sum_{i=2}^T p_i \ |\ X_1 = 0 \right] \leq p_1 + (1-p_1) \cdot 1 = 1 \]
where the inequality follows from the inductive hypothesis applied to $X_2, \dotsc, X_n$.

\end{proof}


Consider the charges added to nodes in $S_i$ in phase $2$.  During phase $2$, vertex $i$ is eligible to be added to $S$ in step $4$.  Write $u_j = |S_i \backslash D(S)|$ for the number of nodes of $S_i$ not dominated on step $j$ of phase $2$.  Then, on each step $j$, $u_j - u_{j+1}$ nodes in $S_i$ are added to $D(S)$, and at least $u_j$ nodes in $G$ are added to $D(S)$ (since this many would be added if $i$ were chosen, and each choice is made greedily).  Thus the total charge added on step $j$ is at most $\frac{u_j - u_{j+1}}{u_j}$.  Since $u_\Delta = 0$, the total charge over all of phase $2$ is at most
$ \sum_{j = 1}^{\Delta-1} \frac{u_j - u_{j+1}}{u_j} \leq \sum_{j = 1}^{\Delta-1} \frac{1}{j} \leq H(\Delta). $
So the expected sum of charges over both phases is at most $1 + H(\Delta)$.

We now turn to show that $\Prob[|S| > 2(2+H(\Delta)) \OPT] < e^{-\OPT}$. \\
We will use the same charging scheme we defined in the main text; it suffices to show that the total charge placed, over all nodes in $G$, is at most $(2 + H(\Delta))\OPT$ with probability at least $1 - e^{-\OPT}$.  Note that our bound on the charges from phase 2 in the analysis of the expected size of $|S|$ holds with probability $1$. it is therefore sufficient to bound the probability that the sum, over all $i$, of the charges placed in phase $1$ of $S_i$ is at most $2\OPT$.

For each node $x$ added to $S$ on line $4$, consider the total number of nodes in $N(x)\backslash D(S)$ that lie in sets $S_i$ that are in phase $1$.  Suppose there are $k$ such nodes, and that $|N(x)\backslash D(S)| = z$.  Then the sum of charges attributed to phase $1$ increases by $k/z$ on this invocation of line $4$.  Also, the probability that any of these $k$ nodes is added to $S$ on the next execution of line $6$ is at least $k/z$, and this would end phase $1$ for at least one set $S_i$.  

We conclude that, if the sum of charges for phase $1$ increases by some $p \in [0,1]$, then with probability $p$ at least one set $S_i$ leaves phase $1$.  Also, no more charges can be attributed to phase $1$ once all sets $S_i$ leave phase $1$, and there are $\OPT$ such sets.  The event that the sum of charges attributed to phase $1$ is greater than $2\OPT$ is therefore dominated by the event that a sequence of Bernoulli random variables $X_1, \dotsc, X_n$, each $X_i$ having mean $p_i$ with $\sum p_i > 2\OPT$, has sum less than $\OPT$.  However, by the multiplicative Chernoff bound (lemma ~\ref{lemma.multchernoff}), this probability is at most
\[ \Prob\left[\sum_{i=1}^n X_i < \OPT \right] = \Prob\left[\sum_{i=1}^n X_i < \frac{1}{2}\E[\sum_{i=1}^n X_i] \right]< e^{-\OPT} \]
as required.

\end{proof}


We end this section by showing that $1^+$-locality is necessary for constructing good local approximation
algorithms.

\begin{theorem}
\label{thm:lowerbndmindom}
For any randomized $1$-local algorithm $A$ for the min dominating set problem, there exists an input instance $G$ for which $\E[|S|] = \Omega(n) \OPT$, where $S$ denotes the output generated by $A$ on input $G$.
\end{theorem}
\begin{proof}
We consider a distribution over input graphs $G = (V,E)$ of size $n$, described by the following construction process.  Choose $n-2$ nodes uniformly at random from $V$ and form a clique on these nodes.  Choose an edge at random from this clique, say $(u,v)$, and remove that edge from the graph.  Finally, let the remaining two nodes be $u'$ and $v'$, and add edges $(u,u')$ and $(v,v')$ to $E$.  By the Yao's minmax principle \cite{Yao77}, it suffices to consider the expected performance of a deterministic $1$-local algorithm on inputs drawn from this distribution.

Note that each such graph has a dominating set of size $2$, namely $\{u,v\}$.
Moreover, any dominating set of $G$ must contain at least one node in $C = \{u,v,u',v'\}$, and hence a $1$-local algorithm must query a node in $C$.  However, if no nodes in $C$ have been queried, then nodes $u$ and $v$ are indistinguishable from other visible unqueried nodes (as they all have degree $n-1$).  Thus, until the algorithm queries a node in $C$, any operation is equivalent to querying an arbitrary unqueried node from $V \backslash \{u',v'\}$.  With high probability, $\Omega(n)$ such queries will be executed before a node in $C$ is selected.

\end{proof}

\subsection{Partial Coverage Problems}
We next study problems in which the goal is not necessarily to cover all nodes in the network, but rather dominate only sections of the network that can be covered efficiently.  We consider two central problems in this domain: the partial dominating set problem and the neighbor collecting problem.

\paragraph{Partial Dominating Set}
In the partial dominating set problem we are given a parameter $\rho \in (0,1]$.  The goal is to find the smallest set $S$ such that $|D(S)| \geq \rho n$. 

We begin with a negative result: for any constant $k$ and any $k$-local algorithm, there are graphs for which the optimal solution has constant size, but with high probability $\Omega(\sqrt{n})$ queries are required to find any $\rho$-partial dominating set. Our example will apply to $\rho = 1/2$, but can be extended to any constant $\rho \in (0,1)$.

\begin{theorem}
\label{thm.partial.lb}
For any randomized $k$-local algorithm $A$ for the partial dominating set problem with $\rho = 1/2$, there exists an input $G$ with optimal partial dominating set $\OPT$ for which 
the expected size of the output returned is $\E[|S|] = \Omega(\sqrt{n}) \cdot |\OPT|$, 
where $S$ denotes the output generated by $A$ on input $G$.
\end{theorem}

\begin{proof}

Fix $n$ and write $r = \frac{n/2-\sqrt{n} -1}{k}$.  We define a distribution over input graphs on $n$ nodes corresponding to the following construction process.  Build two stars, one with $n/2 - \sqrt{n} -1$ leaves and one with $\sqrt{n} -1$ leaves, where the nodes in these stars are chosen uniformly at random.  Let $v$ and $u$ be the roots of these stars, respectively.  Construct $r$ paths, each of length $k+1$, again with the nodes being chosen uniformly at random. Connect one endpoint of each path to a separate leaf of the star rooted at $v$.  Choose one of these $r$ paths and connect its other endpoint to node $u$.  Last, add $\sqrt{n}$ isolated nodes to get the number of nodes equal $n$ in the construction. By the Yao minmax principle \cite{Yao77}, it suffices to consider the expected performance of a deterministic algorithm on a graph chosen from this distribution.

For any such graph, the optimal solution contains two nodes: the root of each star.  We claim that any $k$-local algorithm performs at least $\sqrt{n}$ queries in expectation.  First, if the algorithm does not return the root of the smaller star as part of its solution, then it must return at least $O(\sqrt{n})$ nodes and hence it must use $\Omega(\sqrt{n})$ queries.  On the other hand, suppose that the algorithm does return the root of the smaller star.  Then it must have either traversed the root some node along the path connecting the centers of the stars, or else found a node in the smaller star via a random jump query.  The latter takes $\Omega(\sqrt{n})$ Jump queries, in expectation.  For the former, note that an algorithm cannot distinguish the path connecting the two stars from any other path connected to node $v$, until after a vertex on one of the two paths has been queried.  It would therefore take $\Omega(r) = \Omega(n/k)$ queries in expectation to traverse one of the nodes on the path between the two stars.  We therefore conclude that any algorithm must perform at least $\Omega(\sqrt{n})$ queries in expectation in order to construct an admissible solution.
\end{proof}

Motivated by this lower bound, we consider a bicriteria result: given $\epsilon > 0$, we compare the performance of an algorithm that covers $\rho n$ nodes with the optimal solution that covers $\rho (1+\epsilon) n$ nodes (assuming $\rho(1+\epsilon) \leq 1$).  We show that a modification to Algorithm \ref{alg.mindomset}, in which jumps to uniformly random nodes are interspersed with greedy selections, yields an $O((\rho\epsilon)^{-1}H(\Delta))$ approximation.

\begin{algorithm}
\caption{AlternateRandomAndJump} \label{alg.mindom.partialcoverage}
\begin{algorithmic}[1]
\STATE Initialize $S = \emptyset$.
\WHILE{$D(S) \neq V$}
         \STATE Choose a node $u$ uniformly at random from the graph and add $u$ to $S$.
         \STATE Choose $x \in \arg\max_{v \in N(S)}\{ |N(v) \backslash D(S)| \}$ and add $x$ to $S$.
         \IF{$N(x) \backslash S \neq \emptyset$}
                   \STATE Choose $y \in N(x) \backslash S$ uniformly at random and add $y$ to $S$.
         \ENDIF
\ENDWHILE
\STATE return $S$.
\end{algorithmic}
\end{algorithm} 

\begin{theorem}
\label{thm.partial.alg}
Given any $\rho \in (0,1)$, $\epsilon \in (0, \rho^{-1}-1)$, and set of nodes $\OPT$ with $|D(\OPT)| \geq \rho (1+\epsilon) n$, Algorithm \ref{alg.mindom.partialcoverage} (AlternateRandomAndJump)
returns a set $S$ of nodes with $|D(S)| \geq \rho n$ and $\E[|S|] \leq 3 |\OPT| (\rho\epsilon)^{-1} H(\Delta)$.  
\end{theorem}

\begin{proof}

We apply a modification of the charging argument used in Theorem \ref{thm.mindom}.  Let $\OPT$ be a set of nodes as in the statement of the theorem.  We will partition the nodes of $D(\OPT)$ as follows: for each $i \in \OPT$, choose a set $S_i \subseteq D(\{i\})$ containing $i$, such that the sets $S_i$ form a partition of $D(\OPT)$.  

During the execution of algorithm \ref{alg.mindom.partialcoverage}, we will think of each node in $D(OPT)$ as being marked either as \textbf{Inactive}, \textbf{Active}, or \textbf{Charged}.  At first all nodes in $D(\OPT)$ are marked \textbf{Inactive}.  During the execution of the algorithm, some nodes may have their status changed to \textbf{Active} or \textbf{Charged}.  Once a node becomes \textbf{Active} it never subsequently becomes \textbf{Inactive}, and once a node is marked \textbf{Charged} it remains so for the remainder of the execution.  Specifically, all nodes in $D(\OPT) \cap D(S)$ are always marked \textbf{Charged}, in addition to any nodes that have been assigned a charge by our charging scheme (described below).  Furthermore, for each $i \in \OPT$, the nodes in $S_i$ that are not \textbf{Charged} are said to be \textbf{Active} if $i \in D(S)$; otherwise they are \textbf{Inactive}.

Our charging scheme is as follows.  On each iteration of the loop on lines 2-7, we will either generate a total charge of 0 or of 1.  Consider one such iteration.  Let $u$ be the node that is queried on line $2$ of this iteration.  If $u \not\in D(OPT) \backslash D(S)$ then we will not generate any charge on this iteration.  Suppose instead that $u \in D(OPT) \backslash D(S)$.  If no nodes are \textbf{Active} after $u$ has been queried\footnote{This situation can occur only when $u$ is the only node in $S_i \backslash D(S)$ for some $i$.} then we place a unit of charge on $u$.  Otherwise, let $x$ be the node selected on line 4.  Let $z = |N(x) \backslash D(S)|$ be the number of new nodes dominated by $x$, and let $z'$ be the number of \textbf{Active} nodes.  Let $w = \min\{z,z'\}$.  We will then charge $1/w$ to $w$ different vertices, as follows.  First, we charge $1/w$ to each vertex in $D(OPT) \cap (N(x) \backslash D(S))$ (note that there are at most $w$ such nodes).  If fewer than $w$ nodes have been charged in this way, then charge $1/w$ to (arbitrary) additional \textbf{Active} nodes until a total of $w$ nodes have been charged.  We mark all charged nodes as \textbf{Charged}.

We claim that the total expected charge placed over the course of the algorithm will be $\rho\epsilon|S|/3$.  To see this note that, on each iteration of the algorithm, there are at least $\rho \epsilon n$ nodes in $D(OPT) \backslash D(S)$ (since the algorithm has not yet completed).  Thus, with probability at least $\rho \epsilon$, a node from $D(OPT) \backslash D(S)$ will be chosen on line $2$.  Thus, in expectation, at least a $\rho \epsilon$ fraction of iterations will generate a charge.  Thus, on algorithm termination, the sum of the charges on all vertices is expected to be at least $\rho \epsilon |S|/3$.

Choose some $i \in OPT$ and consider set $S_i$.  We will show that the total expected charge placed on the nodes of $S_i$ during the execution of algorithm $A_2$ is at most $(1+H(\Delta))$.  Since there are $|\OPT|$ such sets, and since only nodes in sets $S_i$ ever receive charge, this will imply that the total expected charge over all nodes is at most $(1+H(\Delta))\OPT$.  We then conclude that $\rho \epsilon |S|/3 \leq (1+H(\Delta))|\OPT|$, completing the proof.

The analysis of the total charge placed on nodes of $S_i$ is similar to the analysis in Theorem \ref{thm.mindom}.  In expectation, a total charge of $1$ will be placed on the nodes of $S_i$ before $i \in D(S)$ (this is phase 1 in the proof of Theorem \ref{thm.mindom}).  After $i \in D(S)$, all nodes in $S_i \backslash D(S)$ are marked \textbf{Active}.  When a node is crawled on line 4, if $k > 0$ nodes in $S_i \backslash D(S)$ are \textbf{Active}, then it must be that $i \in D(S)$ but $i \not\in S$.  Thus, $i$ is a valid choice for the node selected on line $4$.  So, on any such iteration, it must be that the node selected on line 4 dominates at least $k$ new nodes.  We conclude that each node that is charged on this iteration receives a charge of at most $1/k$.

To summarize, if $k$ nodes of $S_i$ are \textbf{Active} on a given iteration, then any nodes in $S_i$ can be charged at most $1/k$ on that iteration.  Since $|S_i| \leq \Delta$, we conclude in the same manner as in Theorem \ref{thm.mindom} that the total charge allocated to nodes in $S_i$, after the first node in $S_i$ becomes \textbf{Active}, is at most $\sum_{k = 1}^{\Delta}\frac{1}{k} = H(\Delta)$.  We conclude that the total expected charge placed on all nodes in $S_i$ is at most $1 + H(\Delta)$, as required.

\end{proof}

\paragraph{The Neighbor Collecting Problem}
We next consider the objective of minimizing the total cost of the selected nodes plus the number of nodes left uncovered: choose a set $S$ of $G$ that minimizes $f(S) = c|S| + |V \backslash D(S)|$ for a given parameter $c > 0$.  This problem is motivated by the Prize-Collecting Steiner Tree problem.

Note that when $c < 1$ the problem reduces to the minimum dominating set problem: it is always worthwhile to cover all nodes.  Assuming $c \geq 1$, the $1^+$-local algorithm for the minimum dominating set problem achieves an $O(c H(\Delta))$ approximation.

\begin{theorem}
\label{thm.nbrcollecting}
For any $c \geq 1$ and set $\OPT$ minimizing $f(\OPT)$, algorithm AlternateRandom returns a set $S$ for which $\E[f(S)] \leq 2c(1+H(\Delta))f(\OPT)$.
\end{theorem}

We also give a lower bound of $\Omega(c,H(\Delta))$ for this problem, illustrating that the
show that the dependency on $c$ is unavoidable. We also show that Theorem \ref{thm.nbrcollecting} cannot be extended to $1$-local algorithms without significant loss.

\begin{proofof}{Theorem \ref{thm.nbrcollecting}}
We have $f(\OPT) = |V-D(\OPT)|+c|\OPT|$ and $f(\OPT \cup \{V - D(\OPT) \}) = c|\OPT \cup  \{V-D(\OPT) \}|  = c|\OPT|+c|V-D(\OPT)| \ge c |T^{*}|$ where $T^{*}$ is a minimum dominating set of the graph.  Next, we know from Theorem \ref{thm.mindom} that 
$|T^{*}| \ge (2(1+H(\Delta)))^{-1} \E[|S|] = (2(1+H(\Delta)))^{-1} c^{-1} f(S)$.

Finally, $f(\OPT) = |V-D(\OPT)| + c |\OPT|$ so 
$c|\OPT|+c|V-D(\OPT)| \le c f(\OPT)$. We conclude that
$f(S) \le 2(1+H(\Delta)) c f(\OPT)$.
\end{proofof}

Since the neighbor-collecting problem contains the minimum dominating set problem as a special case (i.e.\ when $c = 1$), we cannot hope to avoid the dependency on $H(\Delta)$ in the approximation factor in Theorem \ref{thm.nbrcollecting}.  As we next show, the dependence on $c$ in the approximation factor we obtain in Theorem \ref{thm.nbrcollecting} is also unavoidable.

\begin{theorem}
\label{thm.neighbor.c}
For any randomized $k$-local algorithm $A$ for the neighbor-collecting problem where $k =o(n)$, there exists an input instance $G$ for which $\E[f(S)] = \Omega(\max\{c, \log \Delta\}) \cdot f(\OPT)$, where $S$ denotes the output generated by $A$ on input $G$.

\end{theorem}
\begin{proof}
Give $n$ we construct a connected graph on $n$ nodes in the following way. Create two star subgraphs one on $n-\sqrt{n}-2k$ nodes and one on $\sqrt{n}$ nodes.
We connect one arbitrary leaf of the big star subgraph to one arbitrary leaf of the smaller star subgraph.
To complete the construction we choose $k$ spoke nodes from the bigger star subgraph and connect each of them
to one new node of degree one. This gives us a connected graph on $n$ vertices.  Note that OPT is at most $2c +k$ as it can always choose the hubs of the two stars. The worse cost of the min dominating set algorithm 
is at most $(1+2k+2)c$. This happens when the algorithms starts from a spoke in the bigger star component and
need to traverse all $k$ spokes that were assigned one new neighbor. Only after traversing all such nodes 
we move into the spoke of the small star subgraph and then the ub of the smaller star subgraph. Thus the approximation ratio is at least
$\frac{(1+2k+2)c}{2c+k}$. This expression is the biggest (as a function of $k$) for $k=\Theta(c)$. In that case the expression is $\Theta(c)$.
 \end{proof}

Finally, one cannot move from $1^+$-local algorithms to $1$-local algorithms without significant loss: 
every $1$-local algorithm has a polynomial approximation factor.

\begin{theorem}
\label{thm.neighbor.1local}
For any randomized $1$-local algorithm $A$ for the neighbor-collecting problem, there exists an input instance $G$ for which $\E[f(S)] = \Omega(\sqrt{n}/c) \cdot f(\OPT)$, where $S$ denotes the output generated by $A$ on input $G$.
\end{theorem}
\begin{proof}
We construct our graph $G$ as follows.  Build a clique on $n-\sqrt{n}$ vertices and remove one edge $(u,v)$.  Next build a star with $\sqrt{n}-1$ leaves, say with root $r$, and label one of the leaves $v'$.  Finally, add edge $(v,v')$.

For this graph, the set $\{r, v\}$ has cost $2c$.  Consider the set $S$ returned by a $1$-local algorithm; we will show that $S$ will have cost at least $\sqrt{n}$ with high probability.  If $S$ does not include $r$ or $v$ then it must leave $\sqrt{n}$ nodes uncovered (or else contain at least $\sqrt{n}$ vertices), in which case it has cost at least $\sqrt{n}$. So $S$ must contain some node in the star centered at $r$.  A node in the star can be found either via a random query or by querying node $v$.  Since the star contains $\sqrt{n}$ nodes, it would take $\Omega(\sqrt{n})$ random queries to find a node in the star with high probability.  On the other hand, node $v$ is indistinguishable from the other nodes in the $(n-\sqrt{n})$-clique until after it has been queried; it would therefore take $\Omega(n)$ queries to the nodes in the clique to find $v$, again with high probability.  We conclude that the cost of $S$ is at least $\sqrt{n}$ with high probability, as required.
\end{proof}

%% file: LocalInfo-concl.tex
\section{Conclusions}

We presented a model of computation in which algorithms are constrained in the information they have about the input structure, which is revealed over time as expensive exploration decisions are made.  
Our motivation lies in determining whether and how an external user in a network, who cannot make arbitrary queries of the graph structure, can efficiently solve optimization problems in a local manner.
Our results suggest that inherent structural properties of social
networks may be crucial in obtaining strong performance bounds. 

Another implication is that the designer of a network interface, such as an online social network platform, may gain from considering the power and limitations that come with the design choice of how much network topology to reveal to individual users.  On one hand, revealing too little information may restrict natural social processes that users expect to be able to perform, such as searching for potential new connections. 
On the other hand, revealing too much information may raise privacy concerns, or enable unwanted behavior such as automated advertising systems
searching to 
target certain individuals.
Our results suggest that
even minor changes to the structural information made available to a user may have a large impact 
on the class of optimization problems that can be reasonably solved by the user.

%% file: LocalInfo-appendix-PA.tex
\section{Omitted proofs from Section \ref{section.PAapps}}
\label{appendix.PAapps}

\subsection{Degree of the First Fixed $k$ Nodes in Preferential Attachment Networks}
\label{sec.node1degree}

In this section we prove that for any fixed $k$, with high probability, the root node in a preferential attachment network has degree at least $m\sqrt{n}/\log(n)$ and the degree of the $i$'th node, $i \leq k$, is at least $\frac{m\sqrt{n}}{4\log^2(n)}$.  

\begin{lemma}
Let $k$ be fixed. Consider a preferential attachment network on $n$ nodes in which each node generates $m$ links.  Then, with probability at least $1 - o(n^{-1})$,
$\deg(1) \ge \frac{m\sqrt{n}}{\log n}$ and for all $i \leq k:~\deg(i) \ge \frac{m\sqrt{n}}{4\log^2 n}$.
\end{lemma}
\begin{proof}
We will use the notation from Appendix \ref{sec.PA.lemmas}.  In the PA formulation of Appendix \ref{sec.PA.lemmas} we have that $\deg(1)= \sum_{j=2}^{n}\sum_{k=1}^{m}{Y_{k,j}}$ where each of the $Y_{k,j}$s is an i.i.d Bernoulli random variable that gets the value of one with success probability $\frac{w_1}{W_j}$.  From $E_1$ and $E_3$ in Lemma \ref{lem.events}, we have
$$\E(\deg(1)) > \sum_{j=s_0}^{n} \left( m\frac{4/\log(n)\sqrt{n}}{\frac{9}{10}\sqrt{\frac{j}{n}}} \right) = 
\sum_{j=s_0}^{n} \left( m\frac{40}{9} \frac{1}{\log(n)\sqrt{j}} \right).$$
By estimating the sum with an integral we get
$$\E(\deg(1)) > \frac{39m}{9}\frac{\sqrt{n}}{\log(n)} .$$ From the multiplicative Chernoff bound (\ref{lemma.multchernoff}) 
we conclude that with probability bigger than $1-1/n$, $\deg(1) \geq m\frac{\sqrt{n}}{\log(n)}$, as required.
By repeating the same proof and using the lower bound on the wight of a node with small index (event $E_4$ instead of event $E_3$), we get that for any $i \leq k$, with probability bigger than $1-1/n$, $$\deg(i) \geq \frac{m \sqrt{n}}{4\log^2(n)} .$$ The complete claim follows from the union bound.
\end{proof}

%% file: LocalInfo-appendix-mindom.tex

%% file: LocalInfo-appendix-bounds.tex
\section{Concentration Bounds}
\begin{lemma}(\text{Multiplicative Chernoff Bound})
\label{lemma.multchernoff}
Let $X_i$ be $n$ i.i.d. Bernoulli random variables with expectation $\mu$ each. Define $X = \sum_{i=1}^{n}{X_i}$.
Then, \\
For $0 < \lambda <1,~Pr[X < (1-\lambda)\mu n] < \exp(-\mu n\lambda^2/2)$. \\
For $0 < \lambda <1,~Pr[X > (1+\lambda)\mu n] < \exp(-\mu n\lambda^2/4)$. \\
For $\lambda >1,~Pr[X > (1+\lambda)\mu n] < \exp(- \mu n\lambda/2) $.      

\end{lemma}

\begin{lemma}(\text{Additive Chernoff Bound}) 
\label{lemma.addchernoff}
Let $X_i$ be $n$ i.i.d. Bernoulli random variables with expectation $\mu$ each. Define $X = \sum_{i=1}^{n}{X_i}$.
Then, for $ \lambda > 0$, \\
$Pr[X < \mu n - \lambda ] < \exp(-2 {\lambda^2}/n)$. \\
$Pr[X > \mu n + \lambda] < \exp(- 2 {\lambda^2}/n)$.
\end{lemma}

\begin{lemma}(Concentration of Geometric Random Variables) 
\label{lemma.geomchernoff}
Let $Y_i$ be $n$ i.i.d. Geometric random variables with expectation $\mu$ each. Define $Y = \sum_{i=1}^{n}{Y_i}$.
Then,  for $\lambda > 0$, 
\[ Pr[Y > (1+\lambda)\mu n ] \leq \exp(-2\lambda^2 n). \]
\end{lemma}
\begin{proof}
Define $W(n,p)$ to be the a random variable for the number of independent Bernoulli experiments, with bias $p =\frac{1}{\mu}$ each, to get $n$ successes. Denote by $B(t,p)$ a Binomial random variable
on a sequence of $t$ trials and success probability $p$ on each trial. First, by definition, $Y$ has is identically distributed to $W(n,p)$. Next, it easily follows that 
\[ Pr[W(n,p) \geq t ]  =  Pr[B(t,p) \leq n ] \label{eqn:identityWB} \tag{1}~,\] 
\text{ see for example exercise $2.4$ in \cite{DP-09}}. Now set $t = \lceil (1+\lambda)\mu n \rceil$. Then using equation \eqref{eqn:identityWB} and the integrality of $W(n,p)$ we get,
\[ Pr[Y > (1+\lambda)\mu n ] =  Pr[W(n,p) >  (1+\lambda)\mu n] = Pr[W(n,p) \geq  t] = \]
\[ Pr[B(t,p) \leq  n] = Pr[B(t,p) \leq (1+\lambda)n - \lambda n].\] As $E[B(t,p)] = tp \geq (1+\lambda)n$, we get,
 \[ Pr[B(t,p) \leq  n] \leq Pr[B(t,p) \leq E[B(t,p)] - \lambda n].\] Last, by lemma~\ref{lemma.addchernoff} we get,
 \[ Pr[B(t,p) \leq E[B(t,p)] - \lambda n ] \leq \exp( - 2 ((\lambda n)^2 / n) = \exp( - 2\lambda^2 n) .\] 

\end{proof}